\documentclass[a4paper,twocolumn,accepted=2026-05-02]{quantumarticle}
\pdfoutput=1
\usepackage{graphicx}
\usepackage{dcolumn}
\usepackage{bm}
\usepackage{mathrsfs}
\usepackage{braket}
\usepackage{autobreak}
\usepackage{algorithm}
\usepackage{algpseudocode}
\usepackage{xcolor}
\usepackage{amsmath,amssymb,amsfonts,amsthm,mathtools}
\usepackage{quantikz}
\usepackage{empheq}
\usepackage{comment}
\usepackage[caption=false]{subfig}
\usepackage[colorlinks]{hyperref}

\makeatletter
\def\ftype@algorithm{4}
\makeatother

\newtheorem{theorem}{Theorem}

\newtheorem{lemma}{Lemma}

\newtheorem{remark}{Remark}

\newtheorem{claim}{Claim}

\allowdisplaybreaks

\begin{document}

\title{Polynomial time constructive decision algorithm for multivariable quantum signal processing}

\author{Yuki Ito}
\email{yuki.itoh.osaka@gmail.com}
\affiliation{%
Graduate School of Engineering Science, The University of Osaka,
1-3 Machikaneyama, Toyonaka, Osaka 560-8531, Japan
}

\author{Hitomi Mori}
\email{hmori.academic@gmail.com}
\affiliation{%
Graduate School of Engineering Science, The University of Osaka,
1-3 Machikaneyama, Toyonaka, Osaka 560-8531, Japan
}

\author{Kazuki Sakamoto}
\email{kazuki.sakamoto.osaka@gmail.com}
\affiliation{%
Graduate School of Engineering Science, The University of Osaka,
1-3 Machikaneyama, Toyonaka, Osaka 560-8531, Japan
}

\author{Keisuke Fujii}
\email{fujii.keisuke.es@osaka-u.ac.jp}
\affiliation{%
Graduate School of Engineering Science, The University of Osaka,
1-3 Machikaneyama, Toyonaka, Osaka 560-8531, Japan
}
\affiliation{%
Center for Quantum Information and Quantum Biology, The University of Osaka, 
1-2 Machikaneyama, Osaka 560-0043, Japan.
}
\affiliation{RIKEN Center for Quantum Computing (RQC), Hirosawa 2-1, Wako, Saitama 351-0198, Japan}

\maketitle

\begin{abstract}
Quantum signal processing (QSP) and quantum singular value transformation (QSVT) have provided 
a unified framework for understanding many quantum algorithms, including factorization, matrix inversion, and Hamiltonian simulation.
As a multivariable version of QSP,
multivariable quantum signal processing (M-QSP) is proposed.
M-QSP interleaves signal operators corresponding to each variable with signal processing operators, 
which provides an efficient means to perform multivariable polynomial transformations.
However, the necessary and sufficient condition for what types of polynomials can be constructed by M-QSP is unknown. 
In this paper, we propose a classical algorithm to determine whether a given pair of multivariable Laurent polynomials can be implemented by M-QSP, which returns True or False.
As one of the most important properties of this algorithm, 
its returning True is the necessary and sufficient condition.
The proposed classical algorithm runs in polynomial time in the number of variables and signal operators.
Our algorithm also provides 
a constructive method to select the necessary parameters for implementing M-QSP.
These findings offer valuable insights for identifying practical applications of M-QSP.
\end{abstract}

\section{Introduction}

Quantum computers are expected to offer advantages over classical computers
in solving many important problems,
such as factoring~\cite{shor_algorithms_1994}, 
matrix inversion~\cite{harrow_quantum_2009, costa_optimal_2022}, 
and Hamiltonian simulation~\cite{feynman1982simulating, lloyd_universal_1996, low_optimal_2017, gilyen_quantum_2019}.
To improve efficiency in solving these problems, numerous quantum algorithms have been developed. 
In particular, recently proposed quantum signal processing (QSP)~\cite{low_methodology_2016, low_optimal_2017} 
and quantum singular value transformation (QSVT)~\cite{gilyen_quantum_2019} have brought about significant progress.
These include the development of a comprehensive theoretical framework that enhances the understanding of many existing quantum algorithms~\cite{martyn_grand_2021} and improves their performance. 
Moreover, the high extensibility and efficiency of QSP and QSVT have facilitated the creation of new quantum algorithms~\cite{gilyen_quantum_2019, rall2020quantum, lin2020optimal, dong2022ground, rall2021faster, gilyen2022improved, mcardle2022quantum}.

Specifically, QSP performs polynomial transformations of a fixed operator, called a signal operator, interleaving it with controllable parameterized rotation gates, called signal processing operators.
Both operators are represented by $2 \times 2$ unitary matrices, and QSP generates $2 \times 2$ unitary matrices, whose elements correspond to polynomials of a signal variable.
To extend the applicability of QSP from two-dimensional unitary matrices to higher-dimensional ones, which are commonly encountered in most quantum algorithms, a technique known as qubitization was introduced~\cite{low_hamiltonian_2019}.
Qubitization decomposes a higher-dimensional unitary matrix into a direct sum of one- or two-dimensional unitary matrices. 
This decomposition enables QSP to be applied to each of the decomposed unitary matrices.
Combining QSP with qubitization led to QSVT, 
which performs efficient polynomial transformations of the singular values of any matrix embedded in a unitary matrix~\cite{chakraborty2018power, gilyen_quantum_2019, martyn_grand_2021}.
By varying the polynomial utilized in QSVT, QSVT provides systematic solutions to various problems including factoring, matrix inversion, and Hamiltonian simulation~\cite{gilyen_quantum_2019, martyn_grand_2021, lin2020optimal}.

Extending such polynomial transformations to multiple variables,
multivariable quantum signal processing (M-QSP), a multivariable version of QSP, has been proposed~\cite{rossi_multivariable_2022}.
M-QSP achieves multivariable polynomial transformations of several types of signal operators, 
corresponding to each variable, interleaving them with signal processing operators.
Similarly to the single-variable case, these operators are represented by $2 \times 2$ unitary matrices, and 
M-QSP generates $2 \times 2$ unitary matrices consisting of a pair of two multivariable Laurent polynomials $(P,Q)$,
i.e., polynomials that allow both negative and non-negative exponents.
As applications of M-QSP-CDA,
the efficient preparation of initial states required for multivariate Monte Carlo simulations~\cite{mori_efficient_2024} and Coulomb potential calculation~\cite{gomes_multivariable_2024} have been proposed.

To develop the applications of M-QSP,
it is important to clarify which multivariable Laurent polynomials can be constructed by M-QSP;
this is because
the complete characterization of the polynomials that QSP can realize
has enabled a wide range of applications~\cite{low_methodology_2016, gilyen_quantum_2019, martyn_grand_2021}.
Unfortunately, as far as we know, there are no characterizations of such multivariable Laurent polynomials.
Instead of characterizing such polynomials,
the original paper of M-QSP~\cite{rossi_multivariable_2022} presented a characterization of 
the pair of two multivariable Laurent polynomials $(P,Q)$ that can be constructed by M-QSP.
However, subsequent research has proven that its characterization is erroneous~\cite{mori_comment_2024, nemeth_variants_2023}.
To date, the necessary conditions for the pair of multivariable polynomials $(P,Q)$ that can be realized by M-QSP~\cite{mori_comment_2024, nemeth_variants_2023} and the sufficient conditions for a SU(3)-variant of M-QSP~\cite{laneve_multivariate_2025} are known.
However, the necessary and sufficient condition remains unclear.

In this paper, we propose \textit{M-QSP constructivity decision algorithm} (M-QSP-CDA),
a classical algorithm to determine whether a given pair of multivariable Laurent polynomials $(P,Q)$ can be constructed using M-QSP with $n$ signal operators.
That is, M-QSP-CDA provides the necessary and sufficient condition
in the sense that 
if $(P,Q)$ can be constructed by M-QSP with $n$ signal operators, it returns True; 
otherwise, it returns False.

More precisely, this algorithm works as follows: 
first, it checks whether the degree of a given pair of multivariable Laurent polynomials $(P,Q)$ can be reduced by multiplying it by the inverse of a signal processing operator followed by the inverse of a signal operator.
If the degree of $(P,Q)$ cannot be reduced, M-QSP-CDA terminates the process and returns False.
M-QSP-CDA iterates this procedure up to $n$ times, 
or until the degree of the pair of polynomials is reduced to 0.
Finally, it verifies whether the pair of polynomials that is provided after the repetition
can be constructed by M-QSP, which can be decided efficiently from the definition of M-QSP.
If the pair can be constructed by M-QSP with 0 signal operators, 
M-QSP-CDA returns True; 
otherwise, it returns False.
From its construction, whenever this algorithm returns True, such a pair of Laurent polynomials can be constructed by M-QSP.
Thus, this serves as a sufficient condition.
However, 
it is highly nontrivial that this algorithm returning True implies a necessary condition.
This is because there is the possibility that the pair can be implemented by M-QSP 
using the number of signal operators that exceeds the total sum of its degrees of each variable.
To show the necessity,
we show a lemma that
if a given pair of multivariable Laurent polynomials can be constructed by M-QSP,
the pair can be realized with the number of signal operators 
equal to the sum of its degrees of each variable.

The proposed M-QSP-CDA possesses two properties.
First, it runs in polynomial (classical) time in the number of variables and signal operators.
This property theoretically ensures that we can check the proposed necessary and sufficient condition with classical computers efficiently.
In contrast, with finite precision, it is possible that the proposed algorithm does not work well.
The algorithm performs repeated matrix computations, and numerical error may accumulate.
This implies that,
as in Ref.~\cite{haah_product_2019}, 
our algorithm does not guarantee numerical stability in practice.
Nevertheless, theoretically, the algorithm ensures that we can verify, using a finite procedure, whether a given pair of Laurent polynomials can be constructed using M-QSP.
Second, the proposed algorithm provides a constructive method to generate the sequence of signal operators and signal processing operators that implements a given pair of multivariable Laurent polynomials 
if the pair can be constructed by M-QSP.
This fact provides a practical advantage if we need to construct such a pair of Laurent polynomials by M-QSP.
As an application, we also confirm that
the specific pair of multivariable Laurent polynomials in Ref.~\cite{nemeth_variants_2023},
which is provided as a counterexample to the original paper of M-QSP~\cite{rossi_multivariable_2022},
cannot be constructed using M-QSP with an arbitrary number of signal operators.
This example complements the statement of the previous work~\cite{nemeth_variants_2023}.

These results are of great importance 
to specify a more concrete necessary and sufficient condition 
for the characterization of the pair of Laurent polynomials that M-QSP can implement.
Our findings also provide an insight into
seeking practical applications of M-QSP.

The rest of the paper is organized as follows.
In Sec.~\ref{sec:preliminary}, we introduce the definition of Laurent polynomials and M-QSP.
In Sec.~\ref{sec:decision-algorithm}, we describe the proposed algorithm M-QSP-CDA.
Then, we prove that M-QSP-CDA provides the necessary and sufficient condition for the M-QSP constructivity.
Additionally, we explain other properties of M-QSP-CDA.
In Sec.~\ref{sec:example}, as an application of M-QSP-CDA, 
we generalize the result obtained in Ref.~\cite{nemeth_variants_2023}.
Sec.~\ref{sec:conclusion} is devoted to the conclusion and discussion.

\section{Preliminary} \label{sec:preliminary}

In this section, we introduce Laurent polynomials and M-QSP with fixing our notations.
First, let us define Laurent polynomials. 
Given variables $a_1, \dots, a_m$, we say a polynomial $P$ is an $m$-variable Laurent polynomial if
\begin{equation}
    P(a_1, \dots, a_m) =
    \sum_{k_1=-l_1}^{l_1} \dots \sum_{k_m=-l_m}^{l_m} c_{k_1, \dots, k_m} a_1^{k_1} \dots a_m^{k_m}
\end{equation}
holds, where $l_1, \dots, l_m \in \mathbb{Z}_{\ge 0}, \ 
c_{k_1, \dots, k_m} \in \mathbb{C} \ (\lvert k_1 \rvert \le l_1, \dots, \lvert k_m \rvert \le l_m)$.
When the context is clear, 
we write $P(\mathbf{a})$ or simply $P$ instead of $P(a_1, \dots, a_m)$, 
where $\mathbf{a} = (a_1, \dots, a_m)$.
Then, we define the degree of a Laurent polynomial as follows.
For an $m$-variable Laurent polynomial $P \in \mathbb{C}[a_1, a_1^{-1}, \dots, a_m, a_m^{-1}]$, 
if $P$ can be represented as
\begin{equation}
    \begin{split}
        &\quad P(a_1, \dots, a_m) \\
        &=\sum_{k=-l}^{l} P_{a_j^k}(a_1, \dots, a_{j-1}, a_{j+1}, \dots, a_m) a_j^k
        \quad (l \in \mathbb{Z}_{\ge 0}),
    \end{split}
\end{equation}
for $j \in \{1, \dots, m\}$,
then the degree of $P$ with respect to $a_j$, denoted as $\mathrm{deg}_{a_j} P$, is defined as the maximum number $\lvert k \rvert$ satisfying $P_{a_j^k} \neq 0$. 
For $P=0$, we simply define $\mathrm{deg}_{a_j} P=0$ for $j \in \{1, \dots, m\}$.
Also, let $\mathrm{deg} \, P$ denote the sum of the degrees of each variable of $P$, 
i.e., $\mathrm{deg} \, P$ satisfies $\mathrm{deg} \, P = \mathrm{deg}_{a_1} P+\dots+\mathrm{deg}_{a_m} P$.

Next, we explain M-QSP, starting with a description of QSP.
QSP is a technique to perform polynomial transformations of the upper-left element of a signal operator 
\begin{equation}
    \begin{pmatrix}
        x &i \sqrt{1-x^2}\\
        i \sqrt{1-x^2} &x
    \end{pmatrix},
\end{equation}
where $x\in[-1,1]$~\cite{low_methodology_2016, low_optimal_2017}.
That is, QSP consists of signal operators and a carefully chosen sequence of Pauli $Z$ rotations, providing the unitary
\begin{equation}
    \begin{pmatrix}
        P(x) &i Q(x) \sqrt{1-x^2} \\
        i Q^*(x) \sqrt{1-x^2} &P^*(x)
    \end{pmatrix}
\end{equation}
with the upper left elements being transformed as $x \mapsto P(x)$.
When lifting QSP to M-QSP, we change the view of the signal variables from the $x$-picture to the Laurent picture. 
That is, we transform the variables of the signal operators as follows:
\begin{equation}
    x \  \left( x \in [-1,1]  \right) \mapsto 
    \frac{a+a^{-1}}{2} \  \left( a \in \mathbb{T}  \right),
\end{equation}
where $\mathbb{T}$ is the set of complex numbers of modulus one.
In this picture, the polynomials constructed by M-QSP are defined as Laurent polynomials.
Then, we define the $m$-variable M-QSP for a positive integer $m$, based on Refs.~\cite{rossi_multivariable_2022, rossi_modular_2025}.
For $m$-variable Laurent polynomials $P, Q \in \mathbb{C}[a_1, a_1^{-1}, \dots, a_m, a_m^{-1}]$, 
we say that $(P, Q)$ can be constructed by $m$-variable M-QSP in $n$ steps,
where $n$ is a positive integer,
if $(P, Q)$ can be represented as
\begin{equation} \label{eq:definition-M-QSP}
    \begin{pmatrix} 
        P(\mathbf{a}) &Q(\mathbf{a}) \\ 
        -(Q(\mathbf{a}))^* &(P(\mathbf{a}))^* 
    \end{pmatrix}
    = e^{i \phi_0 \sigma_z} \prod_{k=1}^n A(a_{s_k}) e^{i \phi_k \sigma_z}
\end{equation}
where $\sigma_z$ is the Pauli $Z$ operator,
$(\phi_0, \phi_1, \dots, \phi_n) \in \mathbb{R}^{n+1}$, 
$(s_1,\dots, s_n) \in \{1, \dots, m\}^n$ and
the signal operators
\begin{equation}
    A(a_j) \coloneqq
    \begin{pmatrix}
        \frac{a_j+a_j^{-1}}{2} &\frac{a_j-a_j^{-1}}{2} \\
        \frac{a_j-a_j^{-1}}{2} &\frac{a_j+a_j^{-1}}{2}
    \end{pmatrix}
    \quad \left( j \in \{1, \dots, m\} \right)
\end{equation}
are the Pauli $X$ rotations.
Moreover, we state that $(P, Q)$ can be constructed by $m$-variable M-QSP in 0 step,
if $(P, Q)$ satisfies
\begin{equation} \label{eq:deg0}
    \begin{pmatrix} 
        P(\mathbf{a}) &Q(\mathbf{a}) \\ 
        -(Q(\mathbf{a}))^* &(P(\mathbf{a}))^* 
    \end{pmatrix}
    = e^{i \phi_0 \sigma_z},
\end{equation}
which means $P\in\mathbb{T}$ and $Q=0$.
The number of steps $n$ means the number of signal operators in the sequence.
We call $(\phi_0, \phi_1, \dots, \phi_n) \in \mathbb{R}^{n+1}$ angle parameters 
and $(s_1,\dots, s_n) \in \{1, \dots, m\}^n$ index parameters.

\section{M-QSP Constructivity Decision Algorithm (M-QSP-CDA)} \label{sec:decision-algorithm} 

We propose 
a classical algorithm to determine 
in polynomial time 
whether a given pair of $m$-variable Laurent polynomials $(P, Q)$ can be constructed by $m$-variable M-QSP in $n$ steps.
This algorithm is called M-QSP constructivity decision algorithm (M-QSP-CDA).
The procedure of M-QSP-CDA is summarized in the following Algorithm~\ref{alg1}.
Furthermore, we illustrate the process of M-QSP-CDA in Fig.~\ref{fig:decide-func}.

\begin{algorithm*} [t]
    \caption{M-QSP-CDA}
    \label{alg1}
    \begin{algorithmic}[1]
    \State \textbf{Input:} A pair of $m$-variable Laurent polynomials $(P,Q)$, and the number of steps $n$.
    \State \textbf{Output:} The boolean \textsc{M-QSP-CDA}$(P,Q,n)$.
    \State
    \Function {M-QSP-CDA}{$P,Q,n$}
        \If{$n=0$}
            \If{$P \in \mathbb{T}$ and $Q=0$}
                \State \Return True
            \EndIf
        \Else
            \For{$j=1,\dots,m$}
                \State $d_{j} \gets \mathrm{deg}_{a_j} P$
            \EndFor
            \If{$d_{1}+\dots+d_{m} \le n-2$}
                \State \Return \textsc{M-QSP-CDA}$(P,Q,n-2)$
            \ElsIf{$d_{1}+\dots+d_{m} = n$}
                \For{$j=1,\dots,m$}
                    \If{$\exists \varphi_j \in \mathbb{R} \  s.t. \ P_{a_j^{d_{j}}}=e^{2i\varphi_j} Q_{a_j^{d_{j}}}$}
                        \State $P_j \gets e^{-i\varphi_j}\frac{a_j+a_j^{-1}}{2} P - e^{i\varphi_j}\frac{a_j-a_j^{-1}}{2} Q$
                        \State $Q_j \gets e^{i\varphi_j}\frac{a_j+a_j^{-1}}{2} Q - e^{-i\varphi_j}\frac{a_j-a_j^{-1}}{2} P$
                        \State \Return 
                        \textsc{M-QSP-CDA}$(P_j,Q_j,n-1)$
                    \EndIf
                \EndFor
            \EndIf
        \EndIf
        \State \Return False
    \EndFunction
    \end{algorithmic}
\end{algorithm*}

\begin{figure*}[t]
  \begin{center}
    \includegraphics[width=\linewidth]{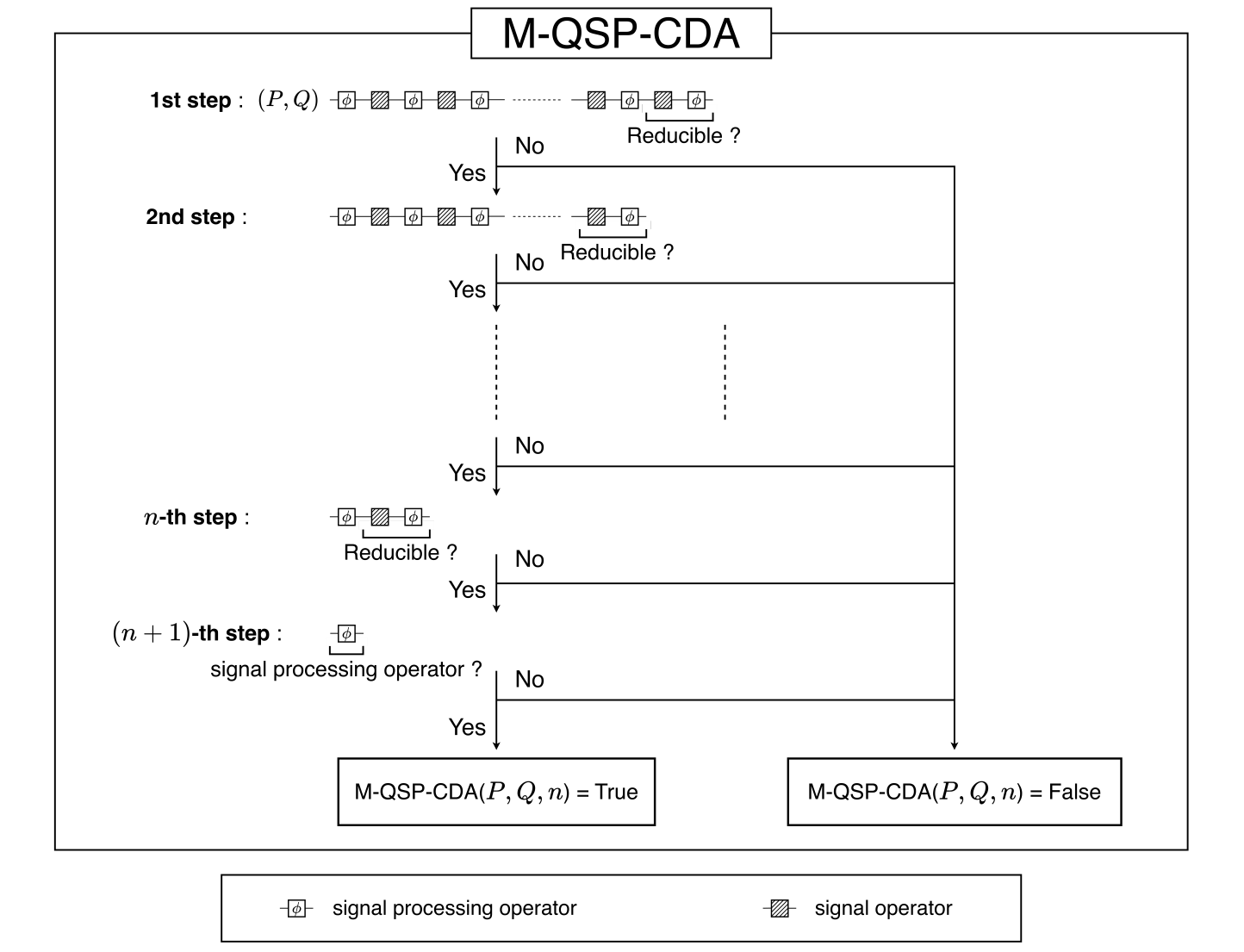}
    \caption{An overview of how M-QSP-CDA shown in Algorithm~\ref{alg1} works.
    M-QSP-CDA recursively verifies if the degree of a given pair of $m$-variable Laurent polynomials can be reduced by multiplying it by the inverse of a signal processing operator followed by the inverse of a signal operator.
    M-QSP-CDA iterates this verification up to $n$ times, or until the degree of the pair of polynomials is reduced to 0.
    Finally, it checks whether the pair of polynomials that is provided after the iteration corresponds to a signal processing operator, based on the definition of $m$-variable M-QSP in 0 step.
    If a given pair of $m$-variable Laurent polynomials $(P,Q)$ passes all verification by M-QSP-CDA,
    it returns True.
    Otherwise, it returns False.}
    \label{fig:decide-func}
  \end{center}
\end{figure*}

\subsection{M-QSP-CDA and the Necessary and Sufficient Condition for the M-QSP Constructivity} 

As one of the most important properties of M-QSP-CDA,
it provides the necessary and sufficient condition for the M-QSP constructivity, as shown in the following Theorem~\ref{thm:m-qsp}.

\begin{theorem} \label{thm:m-qsp}
Given $m$-variable Laurent polynomials $P, Q \in \mathbb{C}[a_1, a_1^{-1}, \dots, a_m, a_m^{-1}]$ 
and a non-negative integer $n$, 
$(P, Q)$ can be constructed by $m$-variable M-QSP in $n$ steps if and only if 
\begin{equation}
    \textsc{M-QSP-CDA}(P,Q,n)= \mathrm{True},
\end{equation}
where the function \textsc{M-QSP-CDA} is defined in Algorithm~\ref{alg1}.
\end{theorem}

\begin{proof}
    See Appendix~\ref{appendix:thm-m-qsp}.
\end{proof}

While detailed proofs are given in the Appendix~\ref{appendix:thm-m-qsp},  in the following, let us first explain how M-QSP-CDA shown in Algorithm~\ref{alg1} works in order to know that M-QSP-CDA gives the necessary and sufficient condition
and provide a sketch of proof later.

The algorithm adopts a recursive structure.
In detail, according to Algorithm~\ref{alg1},
for a given pair of $m$-variable Laurent polynomials,
M-QSP-CDA recursively checks if the degree of the pair can be reduced by multiplying it by the inverse of a signal processing operator followed by the inverse of a signal operator.
More precisely,
for a given pair of $m$-variable Laurent polynomials $(P, Q)$,
it recursively verifies whether there exists an angle $\varphi_j \in \mathbb{R}$ such that 
$P_{a_j^{d_{j}}}=e^{2i\varphi_j} Q_{a_j^{d_{j}}}$ for $j \in \{1, \dots m\}$ in order,
where the degree $d_j$ is equal to $\mathrm{deg}_{a_j} P$, and
$P_{a_j^{d_j}}, Q_{a_j^{d_j}}$ are the coefficients of $a_j^{d_j}$ in $P$ and $Q$ respectively for $j \in \{1, \dots, m\}$.
The relation $P_{a_j^{d_{j}}}=e^{2i\varphi_j} Q_{a_j^{d_{j}}}$ implies
there exists $(P_j, Q_j)$ such that 
$\mathrm{deg}_{a_j} P_j=\mathrm{deg}_{a_j} P-1$, $\mathrm{deg}_{a_j} Q_j=\mathrm{deg}_{a_j} Q-1$, and the decomposition
\begin{equation} \label{eq:relation-(P_j,Q_j)}
    \begin{split}
        &\quad 
        \begin{pmatrix} 
            P_j(\mathbf{a}) &Q_j(\mathbf{a}) \\ 
            -(Q_j(\mathbf{a}))^* &(P_j(\mathbf{a}))^* 
        \end{pmatrix}  \\
        &= 
        \begin{pmatrix} 
            P(\mathbf{a}) &Q(\mathbf{a}) \\ 
            -(Q(\mathbf{a}))^* &(P(\mathbf{a}))^* 
        \end{pmatrix}
        \left(A(a_j)  e^{i \varphi_j \sigma_z} \right)^{-1}.
    \end{split}
\end{equation}
M-QSP-CDA repeats this verification up to $n$ times, or until the degree of the pair of polynomials is reduced to 0,
where $n$ is the given number of steps.
Finally,
M-QSP-CDA examines if
the pair of polynomials $(P^{(0)}, Q^{(0)})$ that is provided after the repetition
meets the condition that
$
\begin{pmatrix}
    P^{(0)}(\mathbf{a}) &Q^{(0)}(\mathbf{a}) \\ 
    -(Q^{(0)}(\mathbf{a}))^* &(P^{(0)}(\mathbf{a}))^*
\end{pmatrix}
$
is a signal processing operator, i.e., $P^{(0)} \in \mathbb{T}$ and $Q^{(0)}=0$.
If a given pair of $m$-variable Laurent polynomials passes the all examination by M-QSP-CDA,
M-QSP-CDA returns True; otherwise, it returns False.

From these procedures,
M-QSP-CDA returning True is sufficient for a given pair to be constructed by $m$-variable M-QSP.
This is because repeatedly inverting the decomposition in Eq.~\eqref{eq:relation-(P_j,Q_j)} recovers $(P, Q)$.
Thus, $\textsc{M-QSP-CDA}(P,Q,n)= \mathrm{True}$ is a sufficient condition for the M-QSP constructivity.

However, it is not trivial to show that $\textsc{M-QSP-CDA}(P,Q,n)= \mathrm{True}$ is also a necessary condition for the constructivity,
from the process of M-QSP-CDA shown in Algorithm~\ref{alg1}.
This is because it is possible that 
a given pair of $m$-variable Laurent polynomials can be implemented by M-QSP 
with much larger steps than the total sum of its degrees of each variable.
To complete the necessity,
we show the following Lemma~\ref{lem:m-n-2}.

\begin{lemma} \label{lem:m-n-2}
    For an integer $n \ge 2$ and given $m$-variable Laurent polynomials $P, Q \in \mathbb{C}[a_1, a_1^{-1}, \dots, a_m, a_m^{-1}]$, 
    it holds that
    \begin{equation} 
        \begin{split}
            &\quad \left\{(P,Q) \middle| 
            \begin{gathered}
                (P,Q) \text{ can be constructed by} \\
                m\text{-variable M-QSP in } (n-2) \text{ steps.}
            \end{gathered}
            \right\} \\
            &=\left\{(P,Q) \middle| 
            \begin{gathered}
                (P,Q) \text{ can be constructed by} \\
                m\text{-variable M-QSP in } n \text{ steps,} \\
                \text{and }
                \mathrm{deg} \, P \le n-2.
            \end{gathered}
             \right\}.
        \end{split}
    \end{equation}
\end{lemma}

\begin{proof}
    See Appendix~\ref{appendix:lem-m-n-2}.
\end{proof}
\noindent
Lemma~\ref{lem:m-n-2} ensures that a given pair of $m$-variable Laurent polynomials implemented by $m$-variable M-QSP can always be constructed with the number of steps equal to the sum of the degrees of each variable.
That is, there are no pairs of Laurent polynomials that cannot be constructed with the same number of steps as the total degree, but can only be constructed by adding more steps. 
Consequently, M-QSP-CDA also provides a necessary condition for the M-QSP constructivity.

Based on this observation, we describe a sketch of the proof for Theorem~\ref{thm:m-qsp}.
Note that the detailed proof is provided in Appendix~\ref{appendix:thm-m-qsp}.
\vspace{1em}
\newline
\textit{Proof sketch of Theorem~\ref{thm:m-qsp}}
\hspace{1em}
We prove by induction on $n \in \mathbb{Z}_{\ge 0}$.
For $n=0$, we can see the definition of the 0-step $m$-variable M-QSP in Eq.~\eqref{eq:deg0} is equivalent to $\textsc{M-QSP-CDA}(P,Q,0)=\text{True}$.

Next, assuming that Theorem~\ref{thm:m-qsp} holds for less than or equal to $n-1(\ge 0)$, we consider the case for $n$.

\noindent
$(\Rightarrow)$
First, assume that $(P, Q)$ can be constructed by $m$-variable M-QSP in $n$ steps.
Let $d_{j} = \mathrm{deg}_{a_j} P$ for $j \in \{1,\dots,m\}$,
and it is obvious that $d_{1}+\dots+d_{m} \le n$.
Next, we show that $\textsc{M-QSP-CDA}(P,Q,n)$ returns True for the cases $d_{1} + \dots + d_{m} < n$ and $d_{1} + \dots + d_{m} = n$, separately.

Suppose that $d_{1} + \dots + d_{m} < n$.
Then, 
due to the parity constraint
$d_1 + \dots + d_m \equiv n \pmod{2}$ (see Lemma~\ref{lem:deg-parity} in Appendix~\ref{appendix:lemma-proof} for detail), 
we have $d_{1} + \dots + d_{m} \le n-2$,
which implies the statement on line 13 of Algorithm~\ref{alg1} holds True.
From line 14 of Algorithm~\ref{alg1}, we obtain $\textsc{M-QSP-CDA}(P,Q,n)=\textsc{M-QSP-CDA}(P,Q,n-2)$.
Furthermore, by applying Lemma~\ref{lem:m-n-2} to $(P,Q)$, 
$(P,Q)$ can be constructed by $m$-variable M-QSP in $(n-2)$ steps.
From the induction hypothesis,
we have $\textsc{M-QSP-CDA}(P,Q,n-2)=\text{True}$,
which leads to $\textsc{M-QSP-CDA}(P,Q,n)=\text{True}$.

Assume that $d_{1} + \dots + d_{m} = n$.
From the assumption that $(P, Q)$ can be constructed by $m$-variable M-QSP in $n$ steps,
focusing on the final step,
we obtain the decomposition of $(P,Q)$ as
\begin{equation} \label{eq:decomposition-n}
    \begin{split}
        &\quad 
        \begin{pmatrix} 
            P(\mathbf{a}) &Q(\mathbf{a}) \\ 
            -(Q(\mathbf{a}))^* &(P(\mathbf{a}))^* 
        \end{pmatrix} \\
        &=
        \begin{pmatrix} 
            P_{s_n}(\mathbf{a}) &Q_{s_n}(\mathbf{a}) \\ 
            -(Q_{s_n}(\mathbf{a}))^* &(P_{s_n}(\mathbf{a}))^* 
        \end{pmatrix} 
        A(a_{s_n}) e^{i \phi_n \sigma_z},
    \end{split}
\end{equation}
where $s_n$ is index of the $n$th signal operator and $\phi_n$ is $n$-th angle parameter of the signal processing operator that implements $(P,Q)$ by $m$-variable M-QSP, appeared in the definition of $m$-variable M-QSP in Eq.~\eqref{eq:definition-M-QSP}.
That is, $s_n$ is the index for the final signal operator and $\phi_n$ is the angle parameter of the final signal processing operator.
Considering the highest-order coefficient of $a_{s_n}$ of Eq.~\eqref{eq:decomposition-n},
we have $P_{a_{s_n}^{d_{s_n}}}=e^{2i\phi_n} Q_{a_{s_n}^{d_{s_n}}}$,
where $P_{a_{s_n}^{d_{s_n}}}, Q_{a_{s_n}^{d_{s_n}}}$ denote the coefficients of $a_{s_n}^{d_{s_n}}$ in $P$ and $Q$ respectively.
Thus,
if the index $j$ for the signal operator on line 16 of Algorithm~\ref{alg1} is equal to $s_n$ at least, 
the statement on line 17 of Algorithm~\ref{alg1} holds True.
In addition, there might be other $j$ than $s_n$ that satisfies 
the condition $P_{a_j^{d_j}}=e^{2i\varphi_j} Q_{a_j^{d_j}}$ on line 17 of Algorithm~\ref{alg1}, 
where $P_{a_j^{d_j}}, Q_{a_j^{d_j}}$ represent the coefficients of $a_j^{d_j}$ in $P$ and $Q$ respectively for $j \in \{1, \dots, m\}$ 
(in the same way as $P_{a_{s_n}^{d_{s_n}}}, Q_{a_{s_n}^{d_{s_n}}}$ case).

Let $s^\prime$ be the smallest index $j \in \{1, \dots, m\}$ among such, and the index used in lines from 18 to 20 of Algorithm~\ref{alg1}.
Then, 
$(P_{s^\prime}, Q_{s^\prime})$ defined by lines 18 and 19 of Algorithm~\ref{alg1} satisfies
$\textsc{M-QSP-CDA}(P,Q,n)=\textsc{M-QSP-CDA}(P_{s^\prime},Q_{s^\prime},n-1)$.
Furthermore, 
since we can obtain $(P_{s^\prime}, Q_{s^\prime})$ by multiplying $(P,Q)$ by $e^{i(-\varphi_{s^\prime}+\pi/2) \sigma_z} A(a_{s^\prime}) e^{-i(\pi/2) \sigma_z}$,
$(P_{s^\prime}, Q_{s^\prime})$ can be constructed by $m$-variable M-QSP in $(n+1)$ steps.
Moreover, since $P_{a_{s^\prime}^{d_{s^\prime}}}=e^{2i\varphi_{s^\prime}} Q_{a_{s^\prime}^{d_{s^\prime}}}$ holds, 
we have $\mathrm{deg} \, P_{s^\prime} < n+1$.
Due to the parity constraint 
$\mathrm{deg} \, P_{s^\prime} \equiv n+1  \pmod{2}$
from Lemma~\ref{lem:deg-parity} in Appendix~\ref{appendix:lemma-proof}, 
we confirm that $\mathrm{deg} \, P_{s^\prime} \le n-1$.
That is, $(P_{s^\prime}, Q_{s^\prime})$ can be constructed by $m$-variable M-QSP in $(n+1)$ steps, 
which satisfies $\mathrm{deg} \, P_{s^\prime} \le n-1$.
Applying Lemma~\ref{lem:m-n-2} to $(P_{s^\prime}, Q_{s^\prime})$,
we obtain that $(P_{s^\prime}, Q_{s^\prime})$ can be constructed by $m$-variable M-QSP in $(n-1)$ steps.
By the induction hypothesis,
we have $\textsc{M-QSP-CDA}(P_{s^\prime}, Q_{s^\prime}, n-1)=\text{True}$,
which implies $\textsc{M-QSP-CDA}(P,Q,n)=\text{True}$.

\noindent
$(\Leftarrow)$
Next, assume that $(P,Q)$ satisfies $\textsc{M-QSP-CDA}(P,Q,n)=\text{True}$.
From the assumption $n-1 \ge 0$, 
$(P,Q)$ satisfies either $\textsc{M-QSP-CDA}(P,Q,n)=\textsc{M-QSP-CDA}(P,Q,n-2)$ or $\textsc{M-QSP-CDA}(P,Q,n)=\textsc{M-QSP-CDA}(P_j, Q_j, n-1)$ for $P_j, Q_j$ defined in lines 18 and 19 of Algorithm~\ref{alg1}.
We show that $(P, Q)$ can be constructed by $m$-variable M-QSP in $n$ steps for each case.

Suppose that $\textsc{M-QSP-CDA}(P,Q,n)=\textsc{M-QSP-CDA}(P,Q,n-2)$.
By the induction hypothesis and $\textsc{M-QSP-CDA}(P,Q,n-2)=\text{True}$, 
$(P, Q)$ can be constructed by $m$-variable M-QSP in $(n-2)$ steps.
Also, since line 13 of Algorithm~\ref{alg1} holds True in this case, 
we have $0 \le d_{1} + \dots + d_{m} \le n-2$, which implies $n \ge 2$.
By $n \ge 2$ and Lemma~\ref{lem:m-n-2}, 
$(P, Q)$ can be implemented by $m$-variable M-QSP in $n$ steps.

Assume that $\textsc{M-QSP-CDA}(P,Q,n)=\textsc{M-QSP-CDA}(P_j, Q_j, n-1)$ for some $j \in \{1, \dots, m\}$.
Then, we obtain $\textsc{M-QSP-CDA}(P_j, Q_j, n-1) = \text{True}$.
Furthermore, by the induction hypothesis, $(P_j, Q_j)$ can be constructed by $m$-variable M-QSP in $(n-1)$ steps.
Since we can obtain the unitary consisting of $(P,Q)$ 
by multiplying $(P_j, Q_j)$ by $A(a_j) e^{2i \varphi_j}$, 
we see that $(P, Q)$ can be implemented by $m$-variable M-QSP in $n$ steps.
\qed

\subsection{The Necessary and Sufficient Condition for the Single-variable Case}

Let us consider the case with $m=1$, i.e., the necessary and sufficient condition for single-variable QSP.
From Theorem~\ref{thm:m-qsp}, $\textsc{M-QSP-CDA} (P,Q,n) = \mathrm{True}$ is also the necessary and sufficient condition even in the case that $(P,Q)$ is a pair of single-variable Laurent polynomials.

According to Ref.~\cite{rossi_multivariable_2022, laneve_multivariate_2025},
there exists another description of the necessary and sufficient condition.
The characterization of the pair of single-variable Laurent polynomials $(P,Q)$ implemented 
by single-variable QSP in $n$ steps
is as follows:
a pair of single-variable Laurent polynomials $(P,Q)$ can be constructed by single-variable QSP in $n$ steps if and only if $(P,Q)$ satisfies
\begin{enumerate}
    \item $\mathrm{deg}(P), \  \mathrm{deg}(Q) \le n$,
    \item $P(a^{-1})=P(a), \ Q(a^{-1})=-Q(a)$,
    \item $P(-a)=(-1)^n P(a), \ Q(-a)=(-1)^n Q(a)$,
    \item For all $a \in \mathbb{T}, \ \lvert P(a) \rvert^2 + \lvert Q(a)\rvert^2=1$.
\end{enumerate}
The conditions 1 to 4 in this characterization are equivalent to $\textsc{M-QSP-CDA} (P,Q,n) = \mathrm{True}$ from Theorem~\ref{thm:m-qsp}.
Specifically, the conditions 2 and 4 in this characterization ensure that 
there exists $\varphi_1 \in \mathbb{R}$ in line 17 of Algorithm~\ref{alg1}.

Note that, in the multivariable case, as far as we know,
there is no algebraic necessary and sufficient condition for the M-QSP constructivity.
A main obstacle to the characterization is 
the multivariable Laurent polynomial constraint appearing in line 17 of Algorithm~\ref{alg1} at each recursive step.
In fact,
the concise algebraic necessary conditions proposed in Ref.~\cite{mori_comment_2024} are not sufficient,
because there is a pair of Laurent polynomials that satisfies the multivariable constraint at some intermediate step but fails at another step.

\subsection{Other Properties of M-QSP-CDA}

We introduce two other properties of M-QSP-CDA.
The first property is that M-QSP-CDA runs in polynomial time
in the number of variables $m$ and signal operators $n$,
as stated in Remark~\ref{remark:complexity}.

\begin{remark} \label{remark:complexity}
    From Theorem~\ref{thm:m-qsp} and Algorithm~\ref{alg1}, 
    we can determine whether $(P,Q)$ can be constructed by $m$-variable M-QSP in $n$ steps in 
    $O \left(nmL \right)$ computational complexity,
    where $L$ is
    \begin{equation}
        L \coloneqq 
        \prod_{j=1}^m \left( 2 \cdot \max\{ \mathrm{deg}_{a_j} P, \mathrm{deg}_{a_j} Q  \} + 1 \right),
    \end{equation}
    which denotes the maximum number of terms in $P$ and $Q$ based on the degrees of $P$ and $Q$.
\end{remark}
\noindent
M-QSP-CDA adopts a recursive structure, where the computational complexity of each step of recursion is at most 
$O \left(mL \right)$,
and the number of recursive calls is at most $n+1$. 
Thus, the computational complexity of M-QSP-CDA is $O \left(nmL \right)$.
This means that, in the absence of numerical errors, 
we can verify the proposed necessary and sufficient condition with classical computers efficiently.
In contrast, with finite precision, it is possible that M-QSP-CDA does not work well.
It recursively determines variable indices and angle parameters step by step through polynomial computations, 
and numerical error may accumulate.
This implies that 
it does not guarantee numerical stability in practice.
Nevertheless, theoretically, M-QSP-CDA provides a finite procedure for checking whether a given pair of Laurent polynomials can be implemented by $m$-variable M-QSP.

Second property is that if M-QSP-CDA returns True for a given pair of $m$-variable Laurent polynomials $(P,Q)$ and the number of steps $n$,
M-QSP-CDA tells us the angle and index parameters to implement $(P,Q)$ using $m$-variable M-QSP,
as described in the following Remark~\ref{remark:params}.
\begin{remark} \label{remark:params}
    If M-QSP-CDA returns True for $(P,Q)$ and $n$, it gives parameters $(\phi_0, \phi_1, \dots, \phi_n) \in \mathbb{R}^{n+1}$ and $(s_1,\dots, s_n) \in \{1, \dots, m\}^n$ in the following way. 
    If $\mathrm{deg} \, P \le n-2$ holds, let
    \begin{equation}
        \phi_{n-1}=\frac{\pi}{2}, \hspace{1em} \phi_{n}=-\frac{\pi}{2}, \hspace{1em} s_{n-1}=s_n=1.
    \end{equation}
    Next, 
    if $\mathrm{deg} \, P = n$ and 
    $\exists \varphi_j \in \mathbb{R} \  s.t. \  P_{a_j^{d_{j}}}=e^{2i\varphi_j} Q_{a_j^{d_{j}}}$ for some $j \in \{1, \dots, m\}$ holds, 
    let
    \begin{equation}
        \phi_n = \varphi_j, \hspace{1em} s_n=j.
    \end{equation}
    Then, we inductively determine $(\phi_1, \dots, \phi_n) \in \mathbb{R}^{n}$ and $(s_1,\dots, s_n) \in \{1, \dots, m\}^n$. 
    Finally, if $n$ is equal to 0,
    $\phi_0 \in \mathbb{R}$ should be chosen to satisfy $P=e^{i \phi_0}$. 
\end{remark}
\noindent
Note that since $A(a_1), A(a_1^{-1}), \dots, A(a_m), A(a_m^{-1})$ are mutually commutative, 
the parameters $(\phi_0, \phi_1, \dots, \phi_n) \in \mathbb{R}^{n+1}$ and $(s_1,\dots, s_n) \in \{1, \dots, m\}^n$ that construct $(P,Q)$ may not be uniquely determined.
In this case, M-QSP-CDA provides one possible choice of the parameters for $(P,Q)$.
Combining this Remark~\ref{remark:params} and Theorem~\ref{thm:m-qsp},
if a pair of $m$-variable Laurent polynomials $(P,Q)$ can be constructed by $m$-variable M-QSP in $n$ steps,
M-QSP-CDA provides a constructive method to select the angle and index parameters to realize $(P,Q)$.

\section{Application to the counterexample in Ref.~\cite{nemeth_variants_2023}} \label{sec:example}

In this section, by using M-QSP-CDA and Lemma~\ref{lem:m-n-2},
we can show that $(P_{2,2}(a,b), Q_{2,2}(a,b))$ in Ref.~\cite{nemeth_variants_2023},
\begin{align}
    \begin{split}
        P_{2,2}(a,b) =
        &\frac{6}{25} \sqrt{\frac{37}{493}} \left[ a^2 b^2 + a^{-2} b^{-2} \right. \\
        &- \left( \frac{122}{37} + \frac{8i}{37} \right) \left( b^2 + b^{-2} \right)  \\
        &+ \left( \frac{114}{37} + \frac{56i}{37} \right) \left( a^{-2} b^2 + a^2 b^{-2} \right) \\
        &+ \left( \frac{362}{111} - \frac{248i}{111} \right) \left( a^2 + a^{-2} \right) \\
        & \left. + \frac{692}{111} - \frac{719i}{222} \right], 
    \end{split} \\
    \begin{split}
        Q_{2,2}(a,b) =
        &\frac{6}{25} \sqrt{\frac{37}{493}} \left[ a^2 b^2 - a^{-2} b^{-2} \right. \\
        &- \left( \frac{122}{37} + \frac{66i}{37} \right) \left( b^2 - b^{-2} \right)  \\
        &+ \left( \frac{56}{37} + \frac{114i}{37} \right) \left( a^{-2} b^2 - a^2 b^{-2} \right) \\
        & \left. + \left( \frac{362}{111} - \frac{418i}{111} \right) \left( a^2 - a^{-2} \right)   \right]
    \end{split}
\end{align}
cannot be constructed by $m$-variable M-QSP in an arbitrary number of steps.
Note that this $(P_{2,2}(a,b), Q_{2,2}(a,b))$ is the counterexample to the necessary and sufficient condition proposed in the original paper of M-QSP~\cite{rossi_multivariable_2022}.
The authors of Ref.~\cite{nemeth_variants_2023} explained that $(P_{2,2}(a,b), Q_{2,2}(a,b))$ cannot be constructed by $m$-variable M-QSP in exactly four steps.
However, there might be the possibility that $(P_{2,2}(a,b), Q_{2,2}(a,b))$ is implemented with larger steps than four, which will be further ruled out by the following argument.

From Lemma~\ref{lem:m-n-2},
a pair of $m$-variable Laurent polynomials $(P,Q)$ that can be implemented by $m$-variable M-QSP
can always be constructed with the number of steps equivalent to the total degree of $P$.
Verifying if $(P,Q)$ can be constructed by $m$-variable M-QSP in $\mathrm{deg} \, P$ steps is necessary and sufficient 
to determine if it can be implemented by $m$-variable M-QSP.
For the case $(P_{2,2}(a,b), Q_{2,2}(a,b))$,
we have $\mathrm{deg} \, P_{2,2}=4$,
and we only need to execute $\textsc{M-QSP-CDA}(P_{2,2}, Q_{2,2}, 4)$.
Focusing on the coefficients of $a^2$ and $b^2$ in $P_{2,2}$ and $Q_{2,2}$,
we obtain $\textsc{M-QSP-CDA}(P_{2,2}, Q_{2,2}, 4)=\mathrm{False}$.
This means $(P_{2,2}(a,b), Q_{2,2}(a,b))$ cannot be constructed by $m$-variable M-QSP in an arbitrary number of steps.
This fact extends the result of Ref.~\cite{nemeth_variants_2023}.

\section{Conclusion and Discussion} \label{sec:conclusion}

In this paper, 
we proposed M-QSP-CDA, a classical algorithm to determine whether a given pair of $m$-variable Laurent polynomials $(P,Q)$ can be constructed by $m$-variable M-QSP in $n$ steps.
As one of the most important properties of M-QSP-CDA, we proved that 
its returning True is the necessary and sufficient condition for the M-QSP constructivity.
The construction of M-QSP-CDA shown in Algorithm~\ref{alg1} obviously provides a sufficient condition, but does not directly mean a necessary condition.
To complete the necessity, 
we showed that 
if a given pair of $m$-variable Laurent polynomials can be constructed by $m$-variable M-QSP,
the pair can be realized with the number of signal operators 
equal to the sum of its degrees of each variable.
As another property of M-QSP-CDA, 
we explained that it runs in polynomial (classical) time in the number of variables and signal operators.
This implies, theoretically,
it provides a finite procedure for verifying
whether a given pair of $m$-variable Laurent polynomials can be implemented by $m$-variable M-QSP.
Note that, since it recursively performs polynomial calculations, 
numerical error may accumulate with finite precision arithmetic, 
and M-QSP-CDA does not guarantee numerical stability in practice.
We also described that
M-QSP-CDA provides a constructive method to choose the necessary angle and index parameters for constructing a given pair of Laurent polynomials by $m$-variable M-QSP,
if the pair can be implemented.
These results provide insight into specifying a more concrete characterization for a pair of $m$-variable Laurent polynomials $(P,Q)$ that can be realized by an $m$-variable M-QSP.

Note that, 
from the counterexample in Sec.~\ref{sec:example},
the M-QSP definition introduced in the original paper~\cite{rossi_multivariable_2022} has a limitation
regarding which multivariable polynomials can be constructed.
Nevertheless, 
this original definition of M-QSP can be regarded as standard, 
since it includes the single-variable case as a special case 
and naturally extends QSP~\cite{gilyen_quantum_2019} by choosing signal operators corresponding to the variables.

One direction for future work is to improve the numerical stability of M-QSP-CDA.
Recent work relating single-variable QSP to the nonlinear Fourier transform~\cite{alexis_quantum_2024, alexis_infinite_2026, ni_fast_2024, laneve_generalized_2025, ni_inverse_2025} may offer useful insights in this direction.
In particular, under certain conditions, 
the nonlinear Fourier transform provides a numerically stable method for sequentially determining each angle parameter 
for constructing a target polynomial~\cite{ni_inverse_2025}.
These results may help improve the numerical stability of M-QSP-CDA 
through an appropriate extension of the nonlinear Fourier transform to the multivariable setting.

Another direction
for future work is to characterize the M-QSP constructivity,
not by an algorithmic condition such as M-QSP-CDA, 
but by a concise algebraic condition as in the single-variable case~\cite{gilyen_quantum_2019}.
A main obstacle is that 
the recursion in Algorithm~\ref{alg1} requires, at each step, the multivariable Laurent polynomial constraint 
$\exists \varphi_j \in \mathbb{R} \  s.t. \ P_{a_j^{d_{j}}}=e^{2i\varphi_j} Q_{a_j^{d_{j}}}$ in line 17.
Understanding this recursive obstruction more deeply may be the key to clarifying the algebraic characterization of $m$-variable M-QSP.
Turning to recent work, 
another possible route toward such an algebraic characterization comes from the adversary-bound approach~\cite{laneve_adversary_2026}.
This line of research shows that the adversary bound for a state conversion problem characterizes the single-variable QSP constructivity, 
while in the multivariate setting it gives sufficient conditions for the constructivity of variants of M-QSP.
Understanding how this adversary-bound approach relates to the decision of M-QSP constructivity produced by M-QSP-CDA 
might help lead to a compact algebraic characterization of M-QSP.

For future work,
it is also important to identify
a complementary $m$-variable Laurent polynomial $Q$
for a given Laurent polynomial $P$,
such that the pair $(P, Q)$ can be constructed by $m$-variable M-QSP.
Recall that M-QSP-CDA returns True if and only if a given pair of $m$-variable Laurent polynomials $(P, Q)$ can be constructed by $m$-variable M-QSP. 
This implies that even if $P$ itself can be implemented by $m$-variable M-QSP, 
M-QSP-CDA will return False if we do not provide a complementary $m$-variable Laurent polynomial $Q$ for $P$.
Thus, 
to ensure that M-QSP-CDA returns True,
it is crucial to identify a complementary $m$-variable Laurent polynomial $Q$ for $P$.
For finding the complementary polynomial in the single-variable setting, 
the theoretically known methods include
a root-finding method~\cite{gilyen_quantum_2019, haah_product_2019} 
and a complex-analysis-based approach~\cite{berntson_complementary_2025}.
However, these methods cannot be extended to a multivariable setting straightforwardly, 
due to the multivariable polynomial constraint in line 17 of Algorithm~\ref{alg1} at each recursive step.
We think it is difficult to find the complementary polynomial by merely inverting or generalizing the multivariable constraint,
because this recursive constraint does not imply a terminal constraint at $n=0$ in line 6 of Algorithm~\ref{alg1}.
Simultaneously satisfying the multivariable condition and the terminal constraint may be crucial for finding a complementary polynomial.

In addition, 
to improve $m$-variable M-QSP usability,
identifying concrete multivariable polynomials constructed by $m$-variable M-QSP is important.
As far as we know, 
it remains unknown which polynomials M-QSP can implement, including polynomials relevant to physical applications, 
such as multivariable Hamiltonian simulation.
Regarding alternative structures that broaden families of polynomials,
one possible direction is to extend the framework by generalizing the signal processing operators from Pauli $Z$ rotations to arbitrary elements of SU(2), as in the generalized QSP case~\cite{motlagh_generalized_2024}.
Another possible extension is to relax commutativity of the operators, as stated in Ref.~\cite{nemeth_variants_2023}.
We expect that addressing these open problems will stimulate the development of the practical application utilizing $m$-variable M-QSP.

\begin{acknowledgments}
    We would like to thank Hayata Morisaki for his valuable comment on the computational complexity of our algorithm.
    We also would like to thank Kaoru Mizuta, Shuntaro Yamamoto, and Nobuyuki Yoshioka for helpful discussions.
    This work is supported by MEXT Quantum Leap Flagship Program (MEXT Q-LEAP) Grant No. JPMXS0120319794 and JST COI-NEXT Grant No. JPMJPF2014.
    YI and KS are also supported by JST SPRING Grant No. JPMJSP2138 and the $\Sigma$ Doctoral Futures Research Grant Program from The University of Osaka.
\end{acknowledgments}

\bibliographystyle{quantum}
\bibliography{m-qsp}

\appendix

\section{Proof of Theorem and Lemma in Section~\ref{sec:decision-algorithm}} \label{appendix:lemma-proof}

We prove the theorem and lemma in Section~\ref{sec:decision-algorithm}, as in Fig.~\ref{fig:appendix-chart}.

\begin{figure*}[t]
  \begin{center}
    \includegraphics[width=\linewidth]{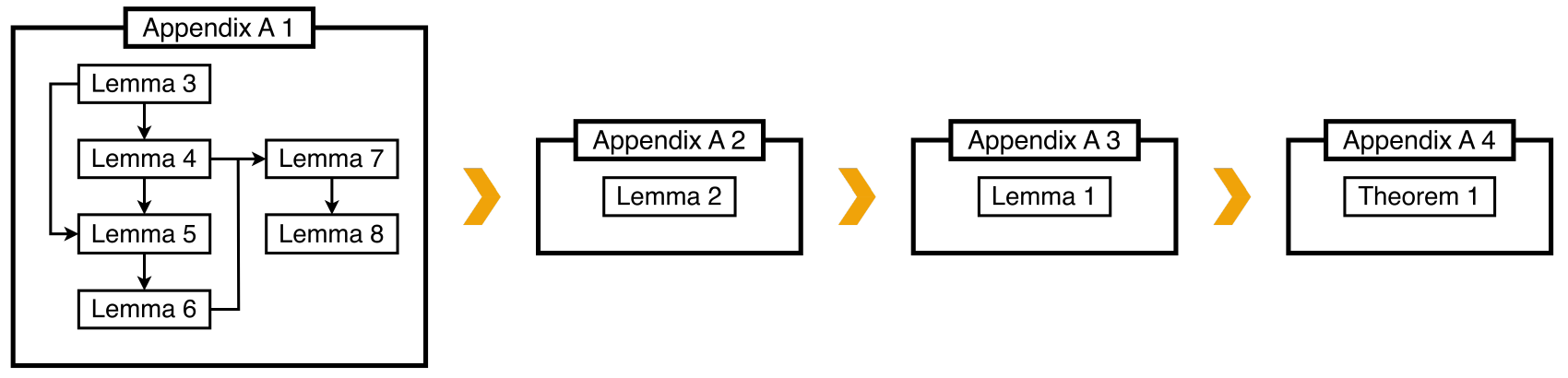}
    \caption{An overview of Appendix~\ref{appendix:lemma-proof}.
    The black arrows indicate the dependencies between Lemmas. 
    The orange arrows represent the flow between each subsection.}
    \label{fig:appendix-chart}
  \end{center}
\end{figure*}

\subsection{Preparation for Proof of Lemma~\ref{lem:m-n-2}} \label{appendix:preparation}

We introduce and prove lemmas necessary for 
showing Lemma~\ref{lem:m-n-2},
which describes the equivalence of the number of steps under the degree condition.
First, the following Lemma~\ref{lem:pi/2-Z} ensures that 
if the sum of the degrees of each variable of $P$ that can be constructed by $m$-variable M-QSP in $l$ steps 
is less than the number of steps $l$, 
the possible values of the angle parameters can be characterized.
As a result, the necessary number of steps to construct $(P, Q)$ is reduced by two. 
Intuitively, 
suppose that a M-QSP protocol starts and ends with a signal operator corresponding to the same variable $a_{s_1}$, 
and $a_{s_1}$ is not chosen in the intermediate signal operators. 
If, in addition, the resulting polynomial $P$ contains no terms involving $a_{s_1}$, 
then the intermediate signal-processing operators must be equal to $\pm I$ or $\pm iZ$.
Thus, the signal operators corresponding to $a_{s_1}$ can be commuted through the sequence 
so that they are moved next to each other and cancel each other.
This reduction technique is used in the proof of Lemma~\ref{lem:m-n-2}.

\begin{lemma} \label{lem:pi/2-Z}
    Let $l$ be an integer and $l \ge 2$. 
    For $(\phi_1, \dots, \phi_l) \in \mathbb{R}^{l}$ and $(s_1,\dots, s_l) \in \{1, \dots, m\}^l$, 
    define
    \begin{equation} \label{eq:P-l-prime-def}
        \begin{pmatrix} 
            P^{(l^\prime)}(\mathbf{a}) &Q^{(l^\prime)}(\mathbf{a}) \\ 
            -(Q^{(l^\prime)}(\mathbf{a}))^* &(P^{(l^\prime)}(\mathbf{a}))^* 
        \end{pmatrix}
        \coloneqq \prod_{k=1}^{l^\prime} A(a_{s_k}) e^{i \phi_k \sigma_z}
    \end{equation}
    for $l^\prime \in \{1, \dots, l\}$.
    If the following four conditions are satisfied:
    \begin{enumerate}
        \item $s_1=s_l$,
        \item For all $k \in \{2, \dots, l-1\}$, $s_k \in \{1, \dots, m\} \setminus \{s_l\}$,
        \item $\mathrm{deg} \, P^{(l-1)}=l-1$,
        \item $\mathrm{deg}_{a_{s_l}} P^{(l)} = 0$,
    \end{enumerate}
    then,
    \begin{equation}
        \phi_1, \dots, \phi_{l-1} \in \frac{\pi}{2} \mathbb{Z}.
    \end{equation}
    Furthermore, $(P^{(l)}, Q^{(l)})$ can be constructed by $m$-variable M-QSP in $(l-2)$ steps.
\end{lemma}

\begin{proof}
    See Appendix~\ref{appendix:lem-pi/2-pi}.
\end{proof}

Next, the following Lemma~\ref{lem:inv-sign-change} describes 
a symmetry with respect to the inverse of the variables.
This implies that, if we need to examine the degrees of $P$ and $Q$, 
it suffices to focus on the coefficients of non-negative powers. 
Note that, in Lemma~\ref{lem:inv-sign-change}, the condition $(\mathrm{ii}^\prime)$ of the revised necessary condition of M-QSP in Ref.~\cite{mori_comment_2024} is extended to $m$ variables.

\begin{lemma} \label{lem:inv-sign-change}
    Let $n$ be a non-negative integer
    and $P, Q \in \mathbb{C}[a_1, a_1^{-1}, \dots, a_m, a_m^{-1}]$ be $m$-variable Laurent polynomials. 
    If $(P, Q)$ can be constructed by $m$-variable M-QSP in $n$ steps, then
    \begin{align}
        P(\mathbf{a}^{-1}) &= P(\mathbf{a}), \label{eq:inv-paritiy-P} \\
        Q(\mathbf{a}^{-1}) &= -Q(\mathbf{a}) \label{eq:inv-paritiy-Q}
    \end{align}
    hold, 
    where $\mathbf{a}^{-1}=(a_1^{-1}, \dots, a_m^{-1})$.
\end{lemma}

\begin{proof}
    Since $(P,Q)$ can be constructed by $m$-variable M-QSP in $n$ steps,
    there exists $(\phi_0, \phi_1, \dots, \phi_n) \in \mathbb{R}^{n+1}$ and $(s_1,\dots, s_n) \in \{1, \dots, m\}^n$ such that we have
    \begin{equation}
        \begin{pmatrix} 
            P(\mathbf{a}) &Q(\mathbf{a}) \\ 
            -(Q(\mathbf{a})^* &(P(\mathbf{a}))^* 
        \end{pmatrix}\\
        =e^{i \phi_0 \sigma_z} \prod_{k=1}^n A(a_{s_k}) e^{i \phi_k \sigma_z}.
    \end{equation}
    If $n=0$,
    we obtain $P=e^{i \phi_0}$ and $Q=0$, 
    and thus Eqs.~\eqref{eq:inv-paritiy-P} and \eqref{eq:inv-paritiy-Q} obviously hold.
    If $n \ge 1$, 
    noting that $A(a_j^{-1})=e^{i (\pi/2) \sigma_z} A(a_j) e^{- i (\pi/2) \sigma_z} \  \left( j \in \{1,\dots,m\} \right)$,
    we can calculate as follows:
    \begin{align}
        &\quad 
        \begin{pmatrix} 
        P(\mathbf{a}^{-1}) &Q(\mathbf{a}^{-1}) \\ 
        -(Q(\mathbf{a}^{-1}))^* &(P(\mathbf{a}^{-1}))^* 
        \end{pmatrix} \notag \\
        &=e^{i \phi_0 \sigma_z} \prod_{k=1}^n A(a_{s_k}^{-1}) e^{i \phi_k \sigma_z} \\
        &=e^{i \phi_0 \sigma_z} 
        \prod_{k=1}^n 
        \left( e^{i (\pi/2) \sigma_z} A(a_{s_k}) e^{- i (\pi/2) \sigma_z} e^{i \phi_k \sigma_z} \right) \\
        &=e^{i (\pi/2) \sigma_z} 
        \left( e^{i \phi_0 \sigma_z} \prod_{k=1}^n A(a_{s_k}) e^{i \phi_k \sigma_z} \right)
        e^{- i (\pi/2) \sigma_z} \\
        &=\begin{pmatrix}
            i &0 \\
            0 &-i
        \end{pmatrix}
        \begin{pmatrix} 
        P(\mathbf{a}) &Q(\mathbf{a}) \\ 
        -(Q(\mathbf{a}))^* &(P(\mathbf{a}))^* 
        \end{pmatrix} 
        \begin{pmatrix}
            -i &0 \\
            0 &i
        \end{pmatrix} \\
        &=
        \begin{pmatrix} 
            P(\mathbf{a}) &-Q(\mathbf{a}) \\ 
            (Q(\mathbf{a}))^* &(P(\mathbf{a}))^* 
        \end{pmatrix}.
    \end{align}
    Thus,
    Eqs.~\eqref{eq:inv-paritiy-P} and \eqref{eq:inv-paritiy-Q} hold for all $n \ge 1$.
\end{proof}

Next, the degrees of the $m$-variable Laurent polynomial $P$ and $Q$ constructed by $m$-variable M-QSP are equal for each variable, as stated in the following Lemma~\ref{lem:relation-deg-P-Q}. 
This implies that it is sufficient to examine only the degrees of $P$ if we need to know the degrees of $P$ and $Q$.

\begin{lemma} \label{lem:relation-deg-P-Q}
    Let $n$ be a non-negative integer and 
    $P, Q \in \mathbb{C}[a_1, a_1^{-1}, \dots, a_m, a_m^{-1}]$ be $m$-variable Laurent polynomials. 
    If $(P,Q)$ can be constructed by $m$-variable M-QSP in $n$ steps, 
    then
    \begin{equation} \label{eq:deg-equal-P-Q}
        \mathrm{deg}_{a_j} P = \mathrm{deg}_{a_j} Q
    \end{equation}
    holds for all $j \in \{1, \dots, m\}$.
\end{lemma}

\begin{proof}
    Note that $\mathrm{deg}_{a_j} 0 = 0$ for all $j \in \{1, \dots, m\}$.
    
    First, we show an important relationship between $P$ and $Q$. 
    Since $(P,Q)$ can be constructed by $m$-variable M-QSP in $n$ steps, 
    the matrix 
    \begin{equation}
        \begin{pmatrix} 
            P(\mathbf{a}) &Q(\mathbf{a}) \\ 
            -(Q(\mathbf{a}))^* &(P(\mathbf{a}))^* 
        \end{pmatrix}
    \end{equation}
    is represented as a product of elements of SU(2). 
    Taking the determinant, for $a_1, \dots, a_m \in \mathbb{T}$,
    we have
    \begin{equation}
        \lvert P(\mathbf{a}) \rvert^2 +  \lvert Q(\mathbf{a}) \rvert^2 = 1.
    \end{equation}
    Also, from Lemma~\ref{lem:inv-sign-change}, 
    which describes the sign change under inversion of the variables,
    for $a_1, \dots, a_m \in \mathbb{T}$,
    \begin{align}
        &(P(\mathbf{a}))^* 
        =( P( \mathbf{a}^{-1} ) )^*
        = P^* (\mathbf{a}), \\
        &(Q(\mathbf{a}))^* 
        =( - Q( \mathbf{a}^{-1} ) )^*
        =- Q^* (\mathbf{a}) 
    \end{align}
    hold, where $P^*, Q^*$ are the Laurent polynomials obtained by taking the complex conjugate of the coefficients of $P, Q$ respectively.
    Therefore, for $a_1, \dots, a_m \in \mathbb{T}$,
    \begin{equation} \label{eq:P-Q-det-1}
        P(\mathbf{a}) (P^*) (\mathbf{a})
        -  Q(\mathbf{a}) (Q^*) (\mathbf{a}) = 1
    \end{equation}
    holds.

    Next, we use Eq.~\eqref{eq:P-Q-det-1} to prove Eq.~\eqref{eq:deg-equal-P-Q}. 
    Let 
    \begin{equation}
        d_{j,P} \coloneqq \mathrm{deg}_{a_j} P, \quad d_{j,Q} \coloneqq \mathrm{deg}_{a_j} Q.
    \end{equation}
    Assume $d_{j,P} < d_{j,Q}$. 
    Since the degree can take non-negative values, we obtain $d_{j,Q} \ge 1$. 
    From the assumption $d_{j,P} < d_{j,Q}$, 
    the coefficient of $a_j^{2 d_{j,Q}}$ in $PP^* - QQ^*$ is $-Q_{a_j^{d_{j,Q}}} (Q_{a_j^{d_{j,Q}}})^*$,
    where $Q_{a_j^{d_{j,Q}}}$ represents the coefficient of $Q$ for $a_j^{d_{j,Q}}$.
    $2 d_{j,Q} > 0$ obtained by $d_{j,Q} \ge 1$ and Eq.~\eqref{eq:P-Q-det-1} imply 
    $-Q_{a_j^{d_{j,Q}}} (Q_{a_j^{d_{j,Q}}})^*=0$. 
    Thus, $Q_{a_j^{d_{j,Q}}}=0$ holds. 
    Considering $Q(\mathbf{a}^{-1}) = -Q(\mathbf{a})$ from Lemma~\ref{lem:inv-sign-change},
    $Q_{a_j^{d_{j,Q}}}=0$ leads to $Q_{a_j^{-d_{j,Q}}}=0$, 
    where $Q_{a_j^{-d_{j,Q}}}$ is the coefficient of $Q$ for $a_j^{-d_{j,Q}}$.
    As a result of $Q_{a_j^{d_{j,Q}}}=Q_{a_j^{-d_{j,Q}}}=0$, we obtain $\mathrm{deg}_{a_j} Q < d_{j,Q}$, which contradicts $d_{j,Q} = \mathrm{deg}_{a_j} Q$. 
    Hence,
    \begin{equation}
        d_{j,P} \ge d_{j,Q}
    \end{equation}
    holds.
    
    By interchanging the roles of $P$ and $Q$ and arguing analogously to the above, 
    we obtain $d_{j,P} \le d_{j,Q}$ by contradiction.
    Thus, $d_{j,P} = d_{j,Q}$ holds.
\end{proof}

The following Lemma~\ref{lem:P-not-0} states that the Laurent polynomial $P$ constructed by $m$-variable M-QSP is non-zero. 
This property plays an important role in the proof of Lemma~\ref{lem:m-n-2} and \ref{lem:deg-parity}.

\begin{lemma} \label{lem:P-not-0}
    Let $n$ be a non-negative integer, and let $P, Q \in \mathbb{C}[a_1, a_1^{-1}, \dots, a_m, a_m^{-1}]$ be $m$-variable Laurent polynomials. 
    If $(P,Q)$ can be constructed by $m$-variable M-QSP in $n$ steps, then
    \begin{equation}
        P(\mathbf{a}) \neq 0
    \end{equation}
    holds.
\end{lemma}

\begin{proof}
    Assume $P(\mathbf{a}) = 0$. 
    Then, it follows that $\mathrm{deg}_{a_j} P = 0$ for $j \in \{1, \dots, m \}$.
    Since $(P,Q)$ can be constructed by $m$-variable M-QSP in $n$ steps, 
    we can apply Lemma~\ref{lem:relation-deg-P-Q}, 
    which describes the degree relation between the pair, to $(P,Q)$.
    We obtain $\mathrm{deg}_{a_j} Q = 0$ for $j \in \{1, \dots, m \}$. 
    Thus, there exists $c \in \mathbb{C}$ such that
    \begin{equation}
         Q(\mathbf{a}) = c.
    \end{equation}
    From Lemma~\ref{lem:inv-sign-change}, 
    which describes the sign change under inversion of the variables,
    it follows that
    \begin{gather}
        Q(\mathbf{a}^{-1}) = -Q(\mathbf{a}) \\
        \therefore \ c = -c \\
        \therefore \ c=0 \\
        \therefore \ Q(\mathbf{a})=0.
    \end{gather}
    Therefore, we have
    \begin{equation} \label{eq:P-Q-0}
        P(\mathbf{a}) = Q(\mathbf{a}) = 0.
    \end{equation}
    Since $(P,Q)$ can be constructed by $m$-variable M-QSP in $n$ steps, 
    \begin{equation}
        \begin{pmatrix} 
            P(\mathbf{a}) &Q(\mathbf{a}) \\ 
            -(Q(\mathbf{a}))^* &(P(\mathbf{a}))^* 
        \end{pmatrix}
    \end{equation}
    is represented as a product of elements of SU(2). 
    Taking the determinant, it follows that
    \begin{equation}
        \lvert P(\mathbf{a}) \rvert^2 +  \lvert Q(\mathbf{a}) \rvert^2 = 1
    \end{equation}
    for $a_1, \dots, a_m \in \mathbb{T}$,
    which contradicts Eq.~\eqref{eq:P-Q-0}.
\end{proof}

The parity of the degree of each variable in the Laurent polynomial $P$ constructed by $m$-variable M-QSP matches the parity of the number of choosing each variable in the sequence of M-QSP,
as in the following Lemma~\ref{lem:deg-parity}. 
This parity constraint allows us to limit the cases to be considered in the proofs of Lemma~\ref{lem:m-n-2} and Theorem~\ref{thm:m-qsp}. 
Moreover, Lemma~\ref{lem:deg-parity} extends the condition $(\mathrm{iii}^\prime)$ of the revised necessary condition of M-QSP from the Ref.~\cite{mori_comment_2024} to $m$ variables.

\begin{lemma} \label{lem:deg-parity}
    Let $n$ be a positive integer and 
    $P, Q \in \mathbb{C}[a_1, a_1^{-1}, \dots, a_m, a_m^{-1}]$ be $m$-variable Laurent polynomials. 
    If $(P, Q)$ can be constructed by $m$-variable M-QSP in $n$ steps, that is, 
    there exists $(\phi_0, \phi_1, \dots, \phi_n) \in \mathbb{R}^{n+1}$ and $(s_1,\dots, s_n) \in \{1, \dots, m\}^n$ such that
    \begin{equation} \label{eq:lem-deg-parity-setting}
        \begin{pmatrix} 
            P(\mathbf{a}) &Q(\mathbf{a}) \\ 
            -(Q(\mathbf{a}))^* &(P(\mathbf{a}))^* 
        \end{pmatrix}
        =e^{i \phi_0 \sigma_z} \prod_{k=1}^n A(a_{s_k}) e^{i \phi_k \sigma_z}.
    \end{equation}
    Then, we obtain
    \begin{equation}
       \mathrm{deg}_{a_j} P \equiv \lvert \{ k \in \{1, \dots, n \} \mid s_k = j\} \rvert \pmod{2}
    \end{equation}
    for all $j \in \{1, \dots, m\}$. 
    In particular,
    \begin{equation} \label{eq:deg_P_equiv_n}
        \mathrm{deg} \, P \equiv n \pmod{2}
    \end{equation}
    holds.
\end{lemma}

\begin{proof}
    For $j \in \{1,\dots,m\}$, let
    \begin{equation}
        c_j \coloneqq \lvert \{ k \in \{1, \dots, n\} \mid s_k=j \}  \rvert.
    \end{equation}
    From Eq.~\eqref{eq:lem-deg-parity-setting},
    \begin{equation}
        P(a_1, \dots, a_{j-1}, -a_j, a_{j+1}, \dots, a_m)
        =(-1)^{c_j} P(\mathbf{a})
    \end{equation}
    holds. 
    Thus, $P$ becomes an even or odd function with respect to each $a_1,\dots,a_m$. 
    Since $P \neq 0$ by Lemma~\ref{lem:P-not-0}, we have
    \begin{equation}
        \mathrm{deg}_{a_j} P \equiv c_j \pmod{2}
    \end{equation}
    for $j \in \{1,\dots,m\}$.
    $c_1+\dots+c_m=n$ implies Eq.~\eqref{eq:deg_P_equiv_n}.
\end{proof}

The following Lemma~\ref{lem:deg-reduction} describes 
the change in the degree for each variable of $P$ which is constructed by $m$-variable M-QSP 
if the number of steps increases by one. 
By focusing on this change, 
we can obtain the degree of the $m$-variable Laurent polynomial $P$ with respect to a specific variable,
which is utilized in the proof of Lemma~\ref{lem:m-n-2}, \ref{lem:deg-sum-equal-step} and Theorem~\ref{thm:m-qsp}.

\begin{lemma} \label{lem:deg-reduction}
    Let $n$ be a non-negative integer and 
    $P, Q \in \mathbb{C}[a_1, a_1^{-1}, \dots, a_m, a_m^{-1}]$ be $m$-variable Laurent polynomials. 
    If $(P,Q)$ can be constructed by $m$-variable M-QSP in $n$ steps and
    \begin{equation} \label{eq:lem-deg-reduction-tilde-def}
        \begin{split}
            &\quad 
            \begin{pmatrix} 
                \tilde{P}(\mathbf{a}) &\tilde{Q}(\mathbf{a}) \\ 
                -(\tilde{Q}(\mathbf{a}))^* &(\tilde{P}(\mathbf{a}))^* 
            \end{pmatrix}\\
            &\coloneqq
            \begin{pmatrix} 
                P(\mathbf{a}) &Q(\mathbf{a}) \\ 
                -(Q(\mathbf{a}))^* &(P(\mathbf{a}))^* 
            \end{pmatrix}
            A(a_s) e^{i \phi \sigma_z} 
        \end{split} 
    \end{equation}
    defines $(\tilde{P}, \tilde{Q})$ with $s \in \{1, \dots, m\}$ and $\phi \in \mathbb{R}$, 
    then we obtain
    \begin{equation} \label{eq:deg-reduction}
        \begin{split}
            &\mathrm{deg}_{a_j} \tilde{P} \\
            =
            &\begin{dcases}
                \mathrm{deg}_{a_j} P -1 \ \mathrm{or} \ \mathrm{deg}_{a_j} P +1 &(\mathrm{if} \ j=s), \\
                \mathrm{deg}_{a_j} P &(\mathrm{if} \ j \in \{1, \dots, m \} \setminus \{s\}).
            \end{dcases}
        \end{split}
    \end{equation}
\end{lemma}

\begin{proof}
    Since $(P,Q)$ can be constructed by $m$-variable M-QSP in $n$ steps, 
    Lemma~\ref{lem:relation-deg-P-Q},
    which describes the degree relation between the pair, implies that
    $\mathrm{deg}_{a_j} P = \mathrm{deg}_{a_j} Q$ for all $j \in \{1, \dots, m\}$.
    From the definition of $(\tilde{P}, \tilde{Q})$ in Eq.~\eqref{eq:lem-deg-reduction-tilde-def},
    we have
    \begin{equation} \label{eq:lem-deg-reduction-deg-limit}
        \begin{dcases}
            \mathrm{deg}_{a_j} \tilde{P} \leq
            \mathrm{deg}_{a_j} P +1 &(\mathrm{if} \ j=s), \\
            \mathrm{deg}_{a_j} \tilde{P} \leq
            \mathrm{deg}_{a_j} P &(\mathrm{if} \ j \in \{1, \dots, m \} \setminus \{s\}).
        \end{dcases}
    \end{equation}
    Additionally,
    from the definition of $(\tilde{P}, \tilde{Q})$ in Eq.~\eqref{eq:lem-deg-reduction-tilde-def},
    we can calculate as follows:
    \begin{align}
        &\quad 
        \begin{pmatrix} 
            P(\mathbf{a}) &Q(\mathbf{a}) \\ 
            -(Q(\mathbf{a}))^* &(P(\mathbf{a}))^* 
        \end{pmatrix} \notag \\
        &= 
        \begin{pmatrix} 
            \tilde{P}(\mathbf{a}) &\tilde{Q}(\mathbf{a}) \\ 
            -(\tilde{Q}(\mathbf{a}))^* &(\tilde{P}(\mathbf{a}))^* 
        \end{pmatrix} 
        e^{-i\phi \sigma_z} A(a_s)^{-1} \\
        &=
        \begin{pmatrix} 
            \tilde{P}(\mathbf{a}) &\tilde{Q}(\mathbf{a}) \\ 
            -(\tilde{Q}(\mathbf{a}))^* &(\tilde{P}(\mathbf{a}))^* 
        \end{pmatrix} 
        e^{i(-\phi + \pi/2) \sigma_z} A(a_s) e^{i(- \pi/2) \sigma_z}. \label{eq:lem-deg-reduction-P-tilde}
    \end{align}
    Moreover,
    since $(\tilde{P}, \tilde{Q})$ can be constructed by $m$-variable M-QSP in $(n+1)$ steps,
    applying Lemma~\ref{lem:relation-deg-P-Q} to $(\tilde{P}, \tilde{Q})$,
    we have
    \begin{equation} \label{eq:deg-tilde-P-equal-tilde-Q}
        \mathrm{deg}_{a_j} \tilde{P} = \mathrm{deg}_{a_j} \tilde{Q}
    \end{equation}
    for $j \in \{1, \dots, m\}$.
    
    First,
    we show the case of $j \in \{1, \dots, m \} \setminus \{s\}$ in Eq.~\eqref{eq:deg-reduction}.
    From Eqs.~\eqref{eq:lem-deg-reduction-P-tilde} and \eqref{eq:deg-tilde-P-equal-tilde-Q},
    we have 
    \begin{equation} \label{eq:lem-deg-reduction-tilde->P}
        \mathrm{deg}_{a_j} P \leq \mathrm{deg}_{a_j} \tilde{P}.
    \end{equation}
    From Eqs.~\eqref{eq:lem-deg-reduction-deg-limit} and \eqref{eq:lem-deg-reduction-tilde->P},
    we obtain
    \begin{equation}
         \mathrm{deg}_{a_j} \tilde{P} = \mathrm{deg}_{a_j} P \quad 
         (j \in \{1, \dots, m \} \setminus \{s\}).
    \end{equation}
    
    Next,
    we show the case of $j=s$ in Eq.~\eqref{eq:deg-reduction}.
    From Eqs.~\eqref{eq:lem-deg-reduction-P-tilde} and \eqref{eq:deg-tilde-P-equal-tilde-Q},
    we have
    \begin{equation} \label{eq:lem-deg-reduction-tilde-P+1}
        \mathrm{deg}_{a_s} P \le \mathrm{deg}_{a_s} \tilde{P} + 1.
    \end{equation}
    According to Eqs.~\eqref{eq:lem-deg-reduction-deg-limit} and \eqref{eq:lem-deg-reduction-tilde-P+1},
    we obtain
    \begin{equation} \label{eq:lem-deg-reduction-tilde-P-deg}
        \mathrm{deg}_{a_s} P - 1 \le \mathrm{deg}_{a_s} \tilde{P} \le \mathrm{deg}_{a_s} P + 1.
    \end{equation}
    Additionally,
    from Lemma~\ref{lem:deg-parity}, 
    which describes the parity constraint between the degree and the number of steps,
    and the definition of $(\tilde{P}, \tilde{Q})$ in Eq.~\eqref{eq:lem-deg-reduction-tilde-def},
    \begin{equation} \label{eq:lem-deg-reduction-tilde-P-mod}
        \mathrm{deg}_{a_s} \tilde{P} \equiv \mathrm{deg}_{a_s} P +1 \pmod{2}
    \end{equation}
    holds.
    From Eqs.~\eqref{eq:lem-deg-reduction-tilde-P-deg},
    and \eqref{eq:lem-deg-reduction-tilde-P-mod},
    we have
    \begin{equation}
         \mathrm{deg}_{a_s} \tilde{P} = \mathrm{deg}_{a_s} P - 1 \ \mathrm{or} \ \mathrm{deg}_{a_s} P + 1.
    \end{equation}
    This completes the proof.
\end{proof}

If the sum of the degrees of each variable in $P$ equals the number of steps $n$, 
the degree of $P$ for each variable is equal to the number of times that variable was chosen,
as stated in the following Lemma~\ref{lem:deg-sum-equal-step}.
Furthermore, the same fact holds for each $(P^{(l)}, Q^{(l)})$ constructed by $m$-variable M-QSP in $l$ steps leading up to $(P, Q)$.

\begin{lemma} \label{lem:deg-sum-equal-step}
    Let $n$ be a positive integer, 
    and for $m$-variable Laurent polynomials $P, Q \in \mathbb{C}[a_1, a_1^{-1}, \dots, a_m, a_m^{-1}]$.
    Suppose $(P, Q)$ can be constructed by $m$-variable M-QSP in $n$ steps, i.e., 
    there exists $(\phi_0, \phi_1, \dots, \phi_n) \in \mathbb{R}^{n+1}$ and $(s_1,\dots, s_n) \in \{1, \dots, m\}^n$ such that
    \begin{equation} 
        \begin{pmatrix} 
            P(\mathbf{a}) &Q(\mathbf{a}) \\ 
            -(Q(\mathbf{a}))^* &(P(\mathbf{a}))^* 
        \end{pmatrix}
        =e^{i \phi_0 \sigma_z} \prod_{k=1}^n A(a_{s_k}) e^{i \phi_k \sigma_z}.
    \end{equation}
    Then, we define
    \begin{equation} \label{eq:P^(l)-Q^(l)-def}
        \begin{pmatrix} 
            P^{(l)}(\mathbf{a}) &Q^{(l)}(\mathbf{a}) \\ 
            -(Q^{(l)}(\mathbf{a}))^* &(P^{(l)}(\mathbf{a}))^* 
        \end{pmatrix}
        \coloneqq e^{i \phi_0 \sigma_z} \prod_{k=1}^{l} A(a_{s_k}) e^{i \phi_k \sigma_z},
    \end{equation}
    for $l \in \{1, \dots, n\}$.
    If $\mathrm{deg} \, P = n$, then
    \begin{gather}
        \mathrm{deg} \, P^{(l)} = l, 
        \label{eq:deg-sum-P-l-equal-l} \\
        \mathrm{deg}_{a_j} P^{(l)} = \lvert \{ k \in \{1, \dots, l\} \mid s_k = j \} \rvert \label{eq:deg-j-P-(l)-equal-choice}
    \end{gather}
    hold for all $(l, j) \in \{1, \dots, n \} \times \{1, \dots,m\}$.
\end{lemma}

\begin{proof}
    First,
    we show Eq.~\eqref{eq:deg-sum-P-l-equal-l} for $l \in \{1, \dots, n\}$.
    If $l=n$,
    Eq.~\eqref{eq:deg-sum-P-l-equal-l} holds by the assumption.
    Next,
    assuming that Eq.~\eqref{eq:deg-sum-P-l-equal-l} holds for $l \ge 2$,
    we show that it also holds for $l-1$.
    According to the definition of $(P^{(l-1)}, Q^{(l-1)})$,
    $(P^{(l-1)}, Q^{(l-1)})$ can be constructed
    by $m$-variable M-QSP in $l-1 \ge 1$ steps.
    Also, we have
    \begin{equation}
        \begin{split}
            &\quad 
            \begin{pmatrix} 
                P^{(l)}(\mathbf{a}) &Q^{(l)}(\mathbf{a}) \\ 
                -(Q^{(l)}(\mathbf{a}))^* &(P^{(l)}(\mathbf{a}))^* 
            \end{pmatrix}\\
            &=
            \begin{pmatrix} 
                P^{(l-1)}(\mathbf{a}) &Q^{(l-1)}(\mathbf{a}) \\ 
                -(Q^{(l-1)}(\mathbf{a}))^* &(P^{(l-1)}(\mathbf{a}))^* 
            \end{pmatrix} 
            A(a_{s_l}) e^{i \phi_l \sigma_z}.
        \end{split}
    \end{equation}
    Therefore,
    from Lemma~\ref{lem:deg-reduction}, 
    which describes how the degrees change when adding one step,
    we obtain
    \begin{equation} \label{eq:deg-P-(l-1)}
        \begin{split}
            &\mathrm{deg}_{a_j} P^{(l-1)} \\
            =
            &\begin{dcases}
                \mathrm{deg}_{a_j} P^{(l)} -1 \ \mathrm{or} \ \mathrm{deg}_{a_j} P^{(l)} +1 &(\mathrm{if} \ j=s_l), \\
                \mathrm{deg}_{a_j} P^{(l)} \\
                (\mathrm{if} \ j \in \{1, \dots, m \} \setminus \{s_l\}).
            \end{dcases}
        \end{split}
    \end{equation}
    Next,
    we assume $\mathrm{deg}_{a_j} P^{(l-1)} = \mathrm{deg}_{a_j} P^{(l)} +1$.
    In this case,
    by the induction hypothesis,
    we have $\mathrm{deg} \, P^{(l)} = l$.
    Thus,
    we see that
    \begin{equation}
        \mathrm{deg} \, P^{(l-1)} = \mathrm{deg} \, P^{(l)} + 1 = l+1.
    \end{equation}
    On the other hand, since $(P^{(l-1)}, Q^{(l-1)})$ can be constructed
    by $m$-variable M-QSP in $(l-1)$ steps,
    \begin{equation}
        \mathrm{deg} \, P^{(l-1)} \le l-1
    \end{equation}
    holds,
    which is a contradiction.
    From Eq.~\eqref{eq:deg-P-(l-1)},
    we have $\mathrm{deg}_{a_j} P^{(l-1)} = \mathrm{deg}_{a_j} P^{(l)} - 1$.
    Hence,
    by the induction hypothesis,
    we have
    \begin{equation}
        \mathrm{deg} \, P^{(l-1)} = \mathrm{deg} \, P^{(l)} - 1 = l-1
    \end{equation}
    and Eq.~\eqref{eq:deg-sum-P-l-equal-l} holds.
    By recursively repeating this process,
    Eq.~\eqref{eq:deg-sum-P-l-equal-l} holds for $l \in \{1, \dots, n\}$.
    
    Next,
    using Eq.~\eqref{eq:deg-sum-P-l-equal-l},
    we show Eq.~\eqref{eq:deg-j-P-(l)-equal-choice}.
    For $l \in \{1, \dots, n\}$,
    according to the definition of $(P^{(l)}, Q^{(l)})$,
    \begin{equation}
        \mathrm{deg}_{a_j} P^{(l)} \le \lvert \{ k \in \{1, \dots, l\} \mid s_k = j \} \rvert
    \end{equation}
    holds for all $j \in \{1, \dots, m\}$.
    From Eq.~\eqref{eq:deg-sum-P-l-equal-l} and
    \begin{equation}
        \sum_{j=1}^m \lvert \{ k \in \{1, \dots, l\} \mid s_k = j \} \rvert = l,
    \end{equation}
    we have
    \begin{equation}
        \mathrm{deg}_{a_j} P^{(l)} = \lvert \{ k \in \{1, \dots, l\} \mid s_k = j \} \rvert.
    \end{equation}
    We obtain the desired result.
\end{proof}

\subsection{Proof of Lemma~\ref{lem:pi/2-Z}} \label{appendix:lem-pi/2-pi}

Next, we prove Lemma~\ref{lem:pi/2-Z}.
The M-QSP sequence satisfying the four conditions in this lemma
implies that intermediate angle parameters must lie in $\frac{\pi}{2}\mathbb{Z}$, and the sequence can be constructed in two fewer steps.
\vspace{1em}
\newline
\textit{Proof of Lemma~\ref{lem:pi/2-Z}}
\hspace{1em}
We prove by induction on $l \in \mathbb{Z}_{\ge 2}$.
\newline
(I) \ For $l=2$,
$P^{(2)}$ in Eq.~\eqref{eq:P-l-prime-def} satisfies
\begin{equation}
    P^{(2)} (\mathbf{a}) 
    =e^{i \phi_2} \left( \frac{\cos \phi_1}{2} a_{s_2}^2 + i \sin \phi_1 + \frac{\cos \phi_1}{2} a_{s_2}^{-2} \right).
\end{equation}
From the condition 4 in Lemma~\ref{lem:pi/2-Z}, 
we obtain $\cos \phi_1 =0$.
Therefore, $\phi_1 \in \left( \frac{\pi}{2}+ \pi \mathbb{Z} \right) \subset  \frac{\pi}{2} \mathbb{Z}$ holds.
Since the condition 1 in Lemma~\ref{lem:pi/2-Z}, i.e., $s_1 = s_2$, 
and $e^{i \phi_1 \sigma_z} A(a_{s_2}) e^{i (-\phi_1) \sigma_z} = A(a_{s_2})^{-1}$,
we have
\begin{align}
    \begin{split}
            &\quad 
            \begin{pmatrix} 
                P^{(2)}(\mathbf{a}) &Q^{(2)}(\mathbf{a}) \\ 
                -(Q^{(2)}(\mathbf{a}))^* &(P^{(2)}(\mathbf{a}))^* 
            \end{pmatrix}\\
            &= A(a_{s_1}) e^{i \phi_1 \sigma_z} A(a_{s_2}) e^{i \phi_2 \sigma_z}
        \end{split} \\
        &=A(a_{s_2}) \left( e^{i \phi_1 \sigma_z} A(a_{s_2}) e^{i (-\phi_1) \sigma_z} \right)
        e^{i (\phi_1 + \phi_2) \sigma_z} \\
        &=A(a_{s_2}) A(a_{s_2})^{-1} e^{i (\phi_1 + \phi_2) \sigma_z} \\
        &=e^{i (\phi_1 + \phi_2) \sigma_z}.
\end{align}
Then, $(P^{(2)}, Q^{(2)})$ can be constructed by $m$-variable M-QSP in $2-2=0$ step.
Hence, Lemma~\ref{lem:pi/2-Z} holds for $l=2$.
\newline
(II) \ Assuming Lemma~\ref{lem:pi/2-Z} holds for $l-1 \ge 2$, we prove it for $l$.
To keep the main line of the proof clear, we state the following claim and use it first.
\begin{claim} \label{claim}
\begin{equation} \label{eq:phi_l-1}
    \phi_{l-1} \in \frac{\pi}{2}  \mathbb{Z}.
\end{equation}
\end{claim}

Assuming Claim~\ref{claim} and using the induction hypothesis, 
we prove that $\phi_1, \dots, \phi_{l-2} \in \frac{\pi}{2} \mathbb{Z}$
and $(P^{(l)}, Q^{(l)})$ can be constructed by $m$-variable M-QSP in $(l-2)$ steps.
Let
\begin{align}
    \tilde{\phi}_k  &\coloneqq \phi_{k} \quad (k \in \{1, \dots, l-3\}), \\
    \tilde{\phi}_{l-2}  &\coloneqq \phi_{l-2}+\phi_{l-1}, \\
    \tilde{\phi}_{l-1}  &\coloneqq 0 \\
    \tilde{s}_k  &\coloneqq s_{k} \quad (k \in \{1, \dots, l-2\}), \\
    \tilde{s}_{l-1}  &\coloneqq s_l.
\end{align}
Then, for $l^\prime \in \{1, \dots, l-1\}$, we define
\begin{equation} \label{eq:PQ-tilde}
    \begin{pmatrix} 
        \tilde{P}^{(l^\prime)}(\mathbf{a}) &\tilde{Q}^{(l^\prime)}(\mathbf{a}) \\ 
        -(\tilde{Q}^{(l^\prime)}(\mathbf{a}))^* &(\tilde{P}^{(l^\prime)}(\mathbf{a}))^* 
    \end{pmatrix}
    \coloneqq \prod_{k=1}^{l^\prime} A(a_{\tilde{s}_k }) e^{i \tilde{\phi}_k } \sigma_z.
\end{equation}
We check that
$(\tilde{\phi}_1, \dots, \tilde{\phi}_{l-1})$ 
and $(\tilde{s}_1,\dots, \tilde{s}_{l-1})$ 
satisfy the conditions 1 to 4 in Lemma~\ref{lem:pi/2-Z}. 
From the definition of $(\tilde{s}_1,\dots, \tilde{s}_{l-1})$ 
and the fact that $(s_1,\dots, s_l)$ satisfies the conditions 1 and 2 in Lemma~\ref{lem:pi/2-Z}, 
we have
\begin{enumerate}
    \item $\tilde{s}_1=\tilde{s}_{l-1}$,
    \item For all $k \in \{2, \dots, l-2\}$, 
    $\tilde{s}_k \in \{1, \dots, m\} \setminus \{\tilde{s}_{l-1} \}$.
\end{enumerate}
Furthermore, 
from the definition of $(\tilde{\phi}_1, \dots, \tilde{\phi}_{l-2})$, $(\tilde{s}_1,\dots, \tilde{s}_{l-2})$ 
and Eq.~\eqref{eq:PQ-tilde}, 
we have
\begin{equation}
    \begin{split}
        &\quad 
        \begin{pmatrix} 
            \tilde{P}^{(l-2)}(\mathbf{a}) &\tilde{Q}^{(l-2)}(\mathbf{a}) \\
            -(\tilde{Q}^{(l-2)}(\mathbf{a}))^* &(\tilde{P}^{(l-2)}(\mathbf{a}))^* 
        \end{pmatrix}\\
        &=
        \begin{pmatrix} 
            P^{(l-2)}(\mathbf{a}) &Q^{(l-2)}(\mathbf{a}) \\ 
            -(Q^{(l-2)}(\mathbf{a}))^* &(P^{(l-2)}(\mathbf{a}))^* 
        \end{pmatrix} 
        e^{i \phi_{l-1} \sigma_z}.
    \end{split}
\end{equation}
Since $e^{i \phi_{l-1} \sigma_z}$ is a diagonal matrix and its diagonal elements are non-zero complex numbers, 
we obtain
\begin{equation} \label{eq:deg-j-tilde-P^(l-2)=P^(l-2)}
    \mathrm{deg}_{a_{j}} \tilde{P}^{(l-2)} = \mathrm{deg}_{a_{j}} P^{(l-2)},
\end{equation}
for all $j \in \{1, \dots, m\}$.
From Eq.~\eqref{eq:P-l-prime-def}, 
$(P^{(l-1)},Q^{(l-1)})$ can be constructed by $m$-variable M-QSP in $l-1 \ge 2$ steps.
Furthermore, from the condition 3 in Lemma~\ref{lem:pi/2-Z}, 
$\mathrm{deg} \, P^{(l-1)}=l-1$ holds.
Then, by applying Lemma~\ref{lem:deg-sum-equal-step}, 
which describes the degrees if $\mathrm{deg} \, P$ equals the number of steps,
to $P^{(l-1)}$, we have
\begin{equation} \label{eq:deg-P^(l-2)=l-2}
    \mathrm{deg} \, P^{(l-2)}=l-2 .
\end{equation}
Eqs.~\eqref{eq:deg-j-tilde-P^(l-2)=P^(l-2)} and \eqref{eq:deg-P^(l-2)=l-2} imply
\begin{enumerate}
    \setcounter{enumi}{2}
    \item $\mathrm{deg} \, \tilde{P}^{(l-2)} = l-2$.
\end{enumerate}
From Eqs.~\eqref{eq:phi_l-1} and \eqref{eq:PQ-tilde}, we have
\begin{align}
    &\quad 
    \begin{pmatrix} 
        \tilde{P}^{(l-1)}(\mathbf{a}) &\tilde{Q}^{(l-1)}(\mathbf{a}) \\ 
        -(\tilde{Q}^{(l-1)}(\mathbf{a}))^* &(\tilde{P}^{(l-1)}(\mathbf{a}))^* 
    \end{pmatrix} \notag \\
    &=
    \begin{aligned}
        &\begin{pmatrix} 
        P^{(l-2)}(\mathbf{a}) &Q^{(l-2)}(\mathbf{a}) \\ 
        -(Q^{(l-2)}(\mathbf{a}))^* &(P^{(l-2)}(\mathbf{a}))^* 
        \end{pmatrix} \\
        &e^{i \phi_{l-1} \sigma_z}
        A(a_{\tilde{s}_{l-1}})
    \end{aligned} \\
    &=
    \begin{aligned}
        &\begin{pmatrix} 
        P^{(l)}(\mathbf{a}) &Q^{(l)}(\mathbf{a}) \\ 
        -(Q^{(l)}(\mathbf{a}))^* &(P^{(l)}(\mathbf{a}))^* 
        \end{pmatrix} \\
        & \left( A(a_{s_{l-1}}) e^{i \phi_{l-1} 
        \sigma_z} A(a_{s_l}) e^{i \phi_l \sigma_z} \right)^{-1}
        e^{i \phi_{l-1} \sigma_z}
        A(a_{s_l})
    \end{aligned} \\
    &=
    \begin{aligned}
        &\begin{pmatrix} 
        P^{(l)}(\mathbf{a}) &Q^{(l)}(\mathbf{a}) \\ 
        -(Q^{(l)}(\mathbf{a}))^* &(P^{(l)}(\mathbf{a}))^* 
        \end{pmatrix} \\
        &e^{-i \phi_l \sigma_z} A(a_{s_l}) ^{-1}
        e^{-i \phi_{l-1} \sigma_z} A(a_{s_{l-1}})^{-1}
        e^{i \phi_{l-1} \sigma_z}
        A(a_{s_l})
    \end{aligned} \\
    &=
    \begin{aligned}
        &\begin{pmatrix} 
        P^{(l)}(\mathbf{a}) &Q^{(l)}(\mathbf{a}) \\ 
        -(Q^{(l)}(\mathbf{a}))^* &(P^{(l)}(\mathbf{a}))^* 
        \end{pmatrix} \\
        &e^{-i (\phi_{l-1} + \phi_l) \sigma_z}
        A(a_{s_{l-1}})^{-1}
        e^{i \phi_{l-1} \sigma_z}
    \end{aligned} \label{eq:P-l-P-tilde-l-1-inv} \\
    &=
    \begin{aligned}
        &\begin{pmatrix} 
        P^{(l)}(\mathbf{a}) &Q^{(l)}(\mathbf{a}) \\ 
        -(Q^{(l)}(\mathbf{a}))^* &(P^{(l)}(\mathbf{a}))^* 
        \end{pmatrix} \\
        &e^{-i (\phi_{l-1} + \phi_l + \pi/2) \sigma_z}
        A(a_{s_{l-1}})
        e^{i (\phi_{l-1} + \pi/2 ) \sigma_z},
    \end{aligned} \label{eq:P-tilde-l-1-P-l}
\end{align}
using the fact that we have
\begin{equation}
    \begin{split}
        &\quad e^{-i \phi_{l-1} \sigma_z} A(a_{s_{l-1}})^{-1} e^{i \phi_{l-1} \sigma_z} \\
        &=
        \begin{dcases}
            A(a_{s_{l-1}}) 
            &\left( \text{if} \  \phi_{l-1}  \in \frac{\pi}{2} + \pi \mathbb{Z} \right), \\
            A(a_{s_{l-1}}^{-1}) 
            &\left( \text{if} \  \phi_{l-1}  \in \pi \mathbb{Z} \right),
        \end{dcases}
    \end{split}
\end{equation}
$A(a_1), A(a_1^{-1}), \dots, A(a_m), A(a_m^{-1})$ are mutually commutative, and 
$A(a_{s_{l-1}})^{-1} = e^{-i (\pi/2) \sigma_z} A(a_{s_{l-1}}) e^{i (\pi/2) \sigma_z}$ holds.
From Eq.~\eqref{eq:P-l-prime-def}, $(P^{(l)},Q^{(l)})$ can be constructed by $m$-variable M-QSP in $l \ge 3$ steps. 
Hence, by applying Lemma~\ref{lem:relation-deg-P-Q},
which describes the degree relation between the pair,
to $(P^{(l)},Q^{(l)})$, 
we see that $\mathrm{deg}_{a_{s_l}} P^{(l)} = \mathrm{deg}_{a_{s_l}} Q^{(l)}$. 
Furthermore, since $(\phi_1, \dots, \phi_l)$ and $(s_1,\dots, s_l)$ satisfy the condition 4 in Lemma~\ref{lem:pi/2-Z}, 
we have $\mathrm{deg}_{a_{s_l}} P^{(l)} = 0$. 
Therefore, $\mathrm{deg}_{a_{s_l}} P^{(l)} = \mathrm{deg}_{a_{s_l}} Q^{(l)} = 0$ holds. 
Consequently, 
from Eq.~\eqref{eq:P-tilde-l-1-P-l}, $s_{l-1} \neq {s_l}$ and $ \tilde{s}_{l-1} = s_l$, 
we have
\begin{enumerate}
    \setcounter{enumi}{3}
    \item $\mathrm{deg}_{a_{\tilde{s}_{l-1}}} \tilde{P}^{(l-1)} = 0$.
\end{enumerate}
Therefore, 
$(\tilde{\phi}_1, \dots, \tilde{\phi}_{l-1})$ 
and $(\tilde{s}_1,\dots, \tilde{s}_{l-1})$ 
satisfy the conditions 1 to 4 in Lemma~\ref{lem:pi/2-Z}.
By the induction hypothesis, we have
\begin{equation}
    \tilde{\phi}_1, \dots, \tilde{\phi}_{l-2} \in \frac{\pi}{2}  \mathbb{Z}.
\end{equation}
From the definition of $( \tilde{\phi}_1, \dots, \tilde{\phi}_{l-2})$ and Eq.~\eqref{eq:phi_l-1},
we obtain
\begin{equation}
    \phi_1, \dots, \phi_{l-1} \in \frac{\pi}{2}  \mathbb{Z}.
\end{equation}
Additionally, from the induction hypothesis, 
$(\tilde{P}^{(l-1)}, \tilde{Q}^{(l-1)})$ can be constructed by $m$-variable M-QSP in $(l-3)$ steps. From Eq.~\eqref{eq:P-l-P-tilde-l-1-inv}, we obtain
\begin{equation}
    \begin{split}
        &
        \begin{pmatrix} 
        P^{(l)}(\mathbf{a}) &Q^{(l)}(\mathbf{a}) \\ 
        -(Q^{(l)}(\mathbf{a}))^* &(P^{(l)}(\mathbf{a}))^* 
        \end{pmatrix} \\
        =&
        \begin{pmatrix} 
            \tilde{P}^{(l-1)}(\mathbf{a}) &\tilde{Q}^{(l-1)}(\mathbf{a}) \\ 
            -(\tilde{Q}^{(l-1)}(\mathbf{a}))^* &(\tilde{P}^{(l-1)}(\mathbf{a}))^* 
        \end{pmatrix} \\
        &e^{-i \phi_{l-1} \sigma_z} A(a_{s_{l-1}}) e^{i (\phi_{l-1}+\phi_l) \sigma_z}.
    \end{split}
\end{equation}
Thus, $(P^{(l)}, Q^{(l)})$ can be constructed by $m$-variable M-QSP in $(l-3)+1=l-2$ steps.

It remains to prove Claim~\ref{claim}.
Let
\begin{equation} \label{eq:P-Q^{[1,l^prime)}-def}
    \begin{pmatrix} 
        P^{[1,l^\prime)} (\mathbf{a}) 
        &Q^{[1,l^\prime)} (\mathbf{a}) \\ 
        -(Q^{[1,l^\prime)} (\mathbf{a}))^* 
        &(P^{[1,l^\prime)} (\mathbf{a}))^* 
    \end{pmatrix}
    \coloneqq e^{i \phi_1 \sigma_z} \prod_{k=2}^{l^\prime} A(a_{s_k}) e^{i \phi_k \sigma_z},
\end{equation}
where $l^\prime \in \{2, \dots, l-1\}$.
Note that the notation $[1,l^\prime)$ is specific to this Lemma~\ref{lem:pi/2-Z}, 
which indicates that the M-QSP sequence starts from the signal processing operator with parameter $\phi_1$.
Now we show
\begin{equation} \label{eq:P-Q^{[1,l-1)}-pure-imag}
    P^{[1,l-1)}, Q^{[1,l-1)} 
    \in  (i \mathbb{R}) [a_1, a_1^{-1}, \dots, a_m, a_m^{-1}].
\end{equation}
We also prove, focusing on coefficients with respect to $a_{s_{l-1}}$,
\begin{equation} \label{eq:P^{[1,l-1)}_{l-1}-equal-Q^{[1,l-1)}_{l-1}}
    P^{[1,l-1) }_{d^\prime}
    = e^{2 i \phi_{l-1}}
    Q^{[1,l-1)}_{d^\prime},
\end{equation}
where $P^{[1,l-1) }_{d^\prime}$ is the coefficient of $P^{[1,l-1)}$ for $a_{s_{l-1}}^{d^{\prime}}$ and
$Q^{[1,l-1) }_{d^\prime}$ is the coefficient of $Q^{[1,l-1)}$ for  $a_{s_{l-1}}^{d^{\prime}}$.
Eq.~\eqref{eq:P-Q^{[1,l-1)}-pure-imag} states that $P^{[1,l-1)}$ and $Q^{[1,l-1)}$ have purely imaginary coefficients.
Eq.~\eqref{eq:P^{[1,l-1)}_{l-1}-equal-Q^{[1,l-1)}_{l-1}} provides a relation between some coefficients of $P^{[1,l-1)}$ and $Q^{[1,l-1)}$.
If Eqs.~\eqref{eq:P-Q^{[1,l-1)}-pure-imag} and \eqref{eq:P^{[1,l-1)}_{l-1}-equal-Q^{[1,l-1)}_{l-1}} hold,
then $e^{2 i \phi_{l-1}}$ must be real, 
and hence
\begin{equation}
    e^{2 i \phi_{l-1}}=1 \ \mathrm{or}\ -1,
\end{equation}
which implies Eq.~\eqref{eq:phi_l-1} and Claim~\ref{claim}.

We first show Eq.~\eqref{eq:P-Q^{[1,l-1)}-pure-imag}.
From the condition 2 in Lemma~\ref{lem:pi/2-Z}, 
it follows that $\mathrm{deg}_{a_{s_l}} P^{[1,l-1)} = \mathrm{deg}_{a_{s_l}} Q^{[1,l-1)} = 0$.
Focusing on the coefficients of $a_{s_l}^2$ and $a_{s_l}^{-2}$ in $P^{(l)}$ defined in Eq.~\eqref{eq:P-l-prime-def} and the condition 4 in Lemma~\ref{lem:pi/2-Z},
\begin{equation} \label{eq:a_j}
    \begin{dcases}
        \begin{aligned}
            &P^{[1,l-1)}(\mathbf{a}) 
            + (P^{[1,l-1)}(\mathbf{a}))^* \\
            &+ Q^{[1,l-1)}(\mathbf{a})
            - (Q^{[1,l-1)}(\mathbf{a}))^* = 0,
        \end{aligned}  \\
        \begin{aligned}
            &P^{[1,l-1)}(\mathbf{a}) 
            + (P^{[1,l-1)}(\mathbf{a}))^* \\
            &- Q^{[1,l-1)}(\mathbf{a})
            + (Q^{[1,l-1)}(\mathbf{a}))^* = 0
        \end{aligned}
    \end{dcases}
\end{equation}
hold. 
From Eq.~\eqref{eq:P-Q^{[1,l^prime)}-def}, 
it is obvious that $(P^{[1,l-1)}, Q^{[1,l-1)})$ can be constructed by $m$-variable M-QSP in $l-2 \ge 1$ steps. 
Then, Lemma~\ref{lem:inv-sign-change},
which describes the sign change under inversion of the variables,
can be applied to $(P^{[1,l-1)}, Q^{[1,l-1)})$,
that is, for $a_1, \dots, a_m \in \mathbb{T}$,
\begin{align}
    (P^{[1,l-1)}(\mathbf{a}))^* 
    &=(P^{[1,l-1)}(\mathbf{a}^{-1}))^*
    = (P^{[1,l-1)})^* (\mathbf{a}), \label{eq:P^{[1,l-1)}-conj}  \\
    (Q^{[1,l-1)}(\mathbf{a}))^* 
    &=(-Q^{[1,l-1)}(\mathbf{a}^{-1}))^*
    =- (Q^{[1,l-1)})^* (\mathbf{a}), \label{eq:Q^{[1,l-1)}-conj}
\end{align}
where $(P^{[1,l-1)})^*, (Q^{[1,l-1)})^*$ are the Laurent polynomials obtained by taking the complex conjugate of the coefficients of $P^{[1,l-1)}, Q^{[1,l-1)}$ respectively.
Inserting Eqs.~\eqref{eq:P^{[1,l-1)}-conj} and \eqref{eq:Q^{[1,l-1)}-conj} into Eq.~\eqref{eq:a_j} yields
\begin{equation}
    \begin{dcases}
        \mathrm{Re}[P^{[1,l-1)}](\mathbf{a}) 
        + \mathrm{Re}[Q^{[1,l-1)}](\mathbf{a}) = 0, \\
        \mathrm{Re}[P^{[1,l-1)}](\mathbf{a}) 
        - \mathrm{Re}[Q^{[1,l-1)}](\mathbf{a}) = 0,
    \end{dcases} 
\end{equation}
where $\mathrm{Re}[P^{[1,l-1)}], \mathrm{Re}[Q^{[1,l-1)}]$ are the Laurent polynomials whose coefficients are the real parts of the coefficients of $P^{[1,l-1)}, Q^{[1,l-1)}$ respectively.
Thus, we have $\mathrm{Re}[P^{[1,l-1)}]=\mathrm{Re}[Q^{[1,l-1)}]=0$.
That is,
$P^{[1,l-1)}$ and $Q^{[1,l-1)}$ have purely imaginary coefficients,
which implies Eq.~\eqref{eq:P-Q^{[1,l-1)}-pure-imag}.

Next, we show Eq.~\eqref{eq:P^{[1,l-1)}_{l-1}-equal-Q^{[1,l-1)}_{l-1}}.
We first see
\begin{equation} \label{eq:deg_{l-1}^{[1,l-2)}}
    \mathrm{deg}_{a_{s_{l-1}}} P^{[1,l-2)}=\mathrm{deg}_{a_{s_{l-1}}} Q^{[1,l-2)}
    =d^{\prime}-1,
\end{equation}
where $d^{\prime} \coloneqq \mathrm{deg}_{a_{s_{l-1}}} P^{[1,l-1)}$.
From Eq.~\eqref{eq:P-l-prime-def},
$(P^{(l-1)},Q^{(l-1)})$ can be constructed by $m$-variable M-QSP in $l-1 \ge 2$ steps.
From the condition 3 in Lemma~\ref{lem:pi/2-Z}, 
we have $\mathrm{deg} \, P^{(l-1)} = l-1$.
Thus, Lemma~\ref{lem:deg-sum-equal-step}, 
which describes the degrees if $\mathrm{deg} \, P$ equals the number of steps,
can be applied, and we obtain
\begin{align}
    \mathrm{deg}_{a_{s_{l-1}}} P^{(l-2)} &= 
    \lvert \{ k \in \{1, \dots, l-2 \} \mid s_k = s_{l-1} \} \rvert, \\
    \mathrm{deg}_{a_{s_{l-1}}} P^{(l-1)} &= 
    \lvert \{ k \in \{1, \dots, l-1 \} \mid s_k = s_{l-1} \} \rvert,
\end{align}
which leads to
\begin{equation} \label{eq:deg_{l-1}-P^(l-2)=P^(l-1)}
    \mathrm{deg}_{a_{s_{l-1}}} P^{(l-2)} = \mathrm{deg}_{a_{s_{l-1}}} P^{(l-1)} - 1.
\end{equation}
Since $s_1 \neq s_{l-1}$ holds from the condition 1 and 2 in Lemma~\ref{lem:pi/2-Z}, 
we have
\begin{align}
    \mathrm{deg}_{a_{s_{l-1}}} P^{(l-2)} &= 
    \lvert \{ k \in \{2, \dots, l-2 \} \mid s_k = s_{l-1} \} \rvert, \\
    \mathrm{deg}_{a_{s_{l-1}}} P^{(l-1)} &= 
    \lvert \{ k \in \{2, \dots, l-1 \} \mid s_k = s_{l-1} \} \rvert.
\end{align}
Therefore, from Eq.~\eqref{eq:P-Q^{[1,l^prime)}-def}, we obtain
\begin{align}
    \mathrm{deg}_{a_{s_{l-1}}} P^{[1,l-2)} &\le \mathrm{deg}_{a_{s_{l-1}}} P^{(l-2)}, \label{eq:deg_{l-1}-P^{[1,l-2)}<=P^(l-2)} \\
    \mathrm{deg}_{a_{s_{l-1}}} P^{[1,l-1)} &\le \mathrm{deg}_{a_{s_{l-1}}} P^{(l-1)} \label{eq:deg_{l-1}-P^{[1,l-1)}<=P^(l-1)}.
\end{align}
Next, we show the converse:
\begin{align}
    \mathrm{deg}_{a_{s_{l-1}}} P^{(l-2)} \le \mathrm{deg}_{a_{s_{l-1}}} P^{[1,l-2)}, \label{eq:deg_{l-1}-P^(l-2)<=P^{[1,l-2)}} \\
    \mathrm{deg}_{a_{s_{l-1}}} P^{(l-1)} \le \mathrm{deg}_{a_{s_{l-1}}} P^{[1,l-1)} \label{eq:deg_{l-1}-P^(l-1)<=P^{[1,l-1)}}.
\end{align}
From the definition of $P^{(l-2)}$ in Eq.~\eqref{eq:P-l-prime-def}, we have
\begin{equation}
    \begin{split}
        &
        \begin{pmatrix} 
            P^{(l-2)}(\mathbf{a}) &Q^{(l-2)}(\mathbf{a}) \\ 
            -(Q^{(l-2)}(\mathbf{a}))^* &(P^{(l-2)}(\mathbf{a}))^* 
        \end{pmatrix} \\
        =&A(a_{s_1})
        \begin{pmatrix} 
            P^{[1,l-2)}(\mathbf{a}) 
            &Q^{[1,l-2)}(\mathbf{a}) \\
            -(Q^{[1,l-2)}(\mathbf{a}))^* 
            &(P^{[1,l-2)}(\mathbf{a}))^*
        \end{pmatrix}.
    \end{split}
\end{equation}
Since $a_{s_{l-1}} \neq a_{s_1}$ from the condition 1 and 2 in Lemma~\ref{lem:pi/2-Z} holds and 
$(P^{[1,l-2)}, Q^{[1,l-2)})$ can be constructed by $m$-variable M-QSP in $l-3 \ge 0$ steps from Eq.~\eqref{eq:P-Q^{[1,l^prime)}-def}, 
we can apply Lemma~\ref{lem:relation-deg-P-Q},
which describes the degree relation between the pair,
to $(P^{[1,l-2)}, Q^{[1,l-2)})$.
Thus, we obtain
\begin{equation} \label{eq:deg_{l-1}-P^{[1,l-2)}=Q^{[1,l-2)}}
    \mathrm{deg}_{a_{s_{l-1}}} P^{[1,l-2)} = \mathrm{deg}_{a_{s_{l-1}}} Q^{[1,l-2)}.
\end{equation}
Therefore,
we have Eq.~\eqref{eq:deg_{l-1}-P^(l-2)<=P^{[1,l-2)}}. 
Similarly, we obtain Eq.~\eqref{eq:deg_{l-1}-P^(l-1)<=P^{[1,l-1)}}.
Then, 
from Eqs.~\eqref{eq:deg_{l-1}-P^{[1,l-2)}<=P^(l-2)}, \eqref{eq:deg_{l-1}-P^{[1,l-1)}<=P^(l-1)}, 
\eqref{eq:deg_{l-1}-P^(l-2)<=P^{[1,l-2)}} and \eqref{eq:deg_{l-1}-P^(l-1)<=P^{[1,l-1)}}, 
we have
\begin{align} 
    \mathrm{deg}_{a_{s_{l-1}}}  P^{[1,l-2)}
    &= \mathrm{deg}_{a_{s_{l-1}}}  P^{(l-2)}, 
    \label{eq:deg_{l-1}-P^{[1,l-2)}-equal-P^(l-2)} \\
    \mathrm{deg}_{a_{s_{l-1}}}  P^{[1,l-1)}
    &= \mathrm{deg}_{a_{s_{l-1}}}  P^{(l-1)}. 
    \label{eq:deg_{l-1}-P^{[1,l-1)}-equal-P^(l-1)}
\end{align}
From Eqs.~\eqref{eq:deg_{l-1}-P^(l-2)=P^(l-1)}, \eqref{eq:deg_{l-1}-P^{[1,l-2)}-equal-P^(l-2)} and \eqref{eq:deg_{l-1}-P^{[1,l-1)}-equal-P^(l-1)}, we obtain
\begin{align}
    \mathrm{deg}_{a_{s_{l-1}}} P^{[1,l-2)} 
    &= \mathrm{deg}_{a_{s_{l-1}}} P^{(l-2)} \\
    &=\mathrm{deg}_{a_{s_{l-1}}} P^{(l-1)} - 1 \\
    &= \mathrm{deg}_{a_{s_{l-1}}} P^{[1,l-1)} -  1 \\
    &=d^{\prime} - 1.
\end{align}
Thus, Eq.~\eqref{eq:deg_{l-1}-P^{[1,l-2)}=Q^{[1,l-2)}} implies Eq.~\eqref{eq:deg_{l-1}^{[1,l-2)}}.

Next,
combining Eq.~\eqref{eq:deg_{l-1}^{[1,l-2)}} with Eq.~\eqref{eq:P-Q^{[1,l^prime)}-def},
we have
\begin{equation}
    \begin{dcases}
        P^{[1,l-1)}_{d^\prime}
        = \frac{e^{i \phi_{l-1}}}{2} 
        \left(
        P^{[1,l-2)}_{d^\prime-1}
        + Q^{[1,l-2)}_{d^\prime-1}
        \right), \\
        Q^{[1,l-1)}_{d^\prime}
        = \frac{e^{-i \phi_{l-1}}}{2} 
        \left(
        P^{[1,l-2)}_{d^\prime-1}
        + Q^{[1,l-2)}_{d^\prime-1}
        \right),
    \end{dcases}
\end{equation}
where $P^{[1,l-2)}_{d^\prime-1}$ is the coefficient of $P^{[1,l-2)}$ for $a_{s_{l-1}}^{d^\prime-1}$ and
$Q^{[1,l-2)}_{d^\prime-1}$ is the coefficient of $Q^{[1,l-2)}$ for $a_{s_{l-1}}^{d^\prime-1}$.
Thus, we obtain Eq.~\eqref{eq:P^{[1,l-1)}_{l-1}-equal-Q^{[1,l-1)}_{l-1}}.
Therefore, we have the desired result.
\qed

\subsection{Proof of Lemma~\ref{lem:m-n-2}} \label{appendix:lem-m-n-2}

Next, we prove Lemma~\ref{lem:m-n-2}.
This lemma states that
a pair of Laurent polynomials $(P,Q)$ can be constructed by $m$-variable M-QSP in $(n-2)$ steps
if and only if 
the pair can be constructed by $m$-variable M-QSP in $n$ steps and satisfies $\mathrm{deg} \, P \le n-2$.
\vspace{1em}
\newline
\textit{Proof of Lemma~\ref{lem:m-n-2}}
\hspace{1em}
$(\subset)$
First, assume that $(P, Q)$ can be constructed by $m$-variable M-QSP in $(n-2)$ steps. 
By appending 
\begin{equation}
    A(a_1) e^{i \frac{\pi}{2} \sigma_z} A(a_1) e^{- i \frac{\pi}{2} \sigma_z} 
    \left( = \begin{pmatrix} 1 &0 \\ 0 &1 \end{pmatrix} \right)
\end{equation}
to the end of a sequence of $(n-2)$ steps of $m$-variable M-QSP consisting of $(P, Q)$, 
$(P, Q)$ can be constructed by  $m$-variable M-QSP in $n$ steps.
$\mathrm{deg} \, P \le n-2$ also holds since $(P, Q)$ can be constructed by $m$-variable M-QSP in $(n-2)$ steps.
\newline
$(\supset)$
Next, we assume that $(P, Q)$ can be constructed by $m$-variable M-QSP in $n$ steps 
and satisfies $\mathrm{deg} \, P \le n-2$. 
Then,
there exists $(\phi_0, \phi_1, \dots, \phi_n) \in \mathbb{R}^{n+1}$ and $(s_1,\dots, s_n) \in \{1, \dots, m\}^n$ such that
\begin{equation}
    \begin{pmatrix} 
        P(\mathbf{a}) &Q(\mathbf{a}) \\ 
        -(Q(\mathbf{a}))^* &(P(\mathbf{a}))^* 
    \end{pmatrix}
    =e^{i \phi_0 \sigma_z} \prod_{k=1}^n A(a_{s_k}) e^{i \phi_k \sigma_z}.
\end{equation}
For $l \in \{0,1, \dots, n\}$, let
\begin{equation} \label{eq:P-(l)-def}
    \begin{pmatrix} 
        P^{(l)}(\mathbf{a}) 
        &Q^{(l)}(\mathbf{a}) \\
        -(Q^{(l)}(\mathbf{a}))^* 
        &(P^{(l)}(\mathbf{a}))^* 
        \end{pmatrix}
        \coloneqq e^{i \phi_0 \sigma_z} \prod_{k=1}^l A(a_{s_k}) e^{i \phi_k \sigma_z}.
\end{equation}
Furthermore, for $j \in \{1, \dots, m\}$, we define
\begin{equation}
    c(j,l) \coloneqq \lvert \{ k \in \{1, \dots, l\} \mid s_k=j \} \rvert,
\end{equation}
where $c(j,0)=0$.

First, we show the existence of
\begin{equation}
    \mathfrak{r} \coloneqq \min \left\{ l \in \{0,1, \dots, n\} \middle| 
    \begin{gathered}
        \exists j \in \{1, \dots, m\} \ s.t. \ \\
        \mathrm{deg}_{a_j} P^{(l)} \neq c(j,l)
    \end{gathered}
    \right\}.
\end{equation}
From
\begin{gather}
    \mathrm{deg} \, P \le n-2, \\
    \mathrm{deg}_{a_j} P \le c(j,n) \quad (j \in \{1, \dots, m\})
\end{gather}
and
\begin{equation}
    \sum_{j=1}^m c(j,n) = n,
\end{equation}
we observe that
\begin{equation}
    \exists j \in \{1, \dots, m\} \ s.t. \ \mathrm{deg}_{a_j} P \neq c(j,n).
\end{equation}
Thus, $\mathfrak{r}$ exists.

Next, we show
\begin{equation} \label{eq:r-ge-2}
    \mathfrak{r} \ge 2.
\end{equation}
In fact, considering that $P^{(0)}=e^{i \phi_0}$ and $c(j,0)=0$ for all $j \in \{1, \dots, m\}$, it follows that
\begin{equation}
    \forall j \in \{1, \dots, m\}, \ \mathrm{deg}_{a_j} P^{(0)} = c(j,0)
\end{equation}
holds, which means $\mathfrak{r} \neq 0$.
From $P^{(1)}=e^{i (\phi_0 + \phi_1)} (a_{s_1}+a_{s_1}^{-1})/2$ and the definition of $c(j,1)$,
we obtain
\begin{equation}
    \forall j \in \{1, \dots, m\}, \  \mathrm{deg}_{a_j} P^{(1)} = c(j,1),
\end{equation}
which implies $\mathfrak{r} \neq 1$.
Therefore, we see that Eq.~\eqref{eq:r-ge-2}.

Note that, due to the minimality of $\mathfrak{r}$, we have
\begin{equation} \label{eq:deg_r}
    \begin{split}
        &\forall (l, j) \in \{0, 1, \dots, \mathfrak{r} -1 \} \times \{1, \dots, m \}, \\
        &\mathrm{deg}_{a_j} P^{(l)} = c(j,l)
    \end{split}
\end{equation}
holds.
Next, we show
\begin{equation}
    \mathrm{deg}_{a_{s_{\mathfrak{r}}}} P^{(\mathfrak{r})} 
    < c(s_{\mathfrak{r}}, \mathfrak{r}),
\end{equation}
utilizing Eq.~\eqref{eq:deg_r}.
From Eqs.~\eqref{eq:P-(l)-def} and \eqref{eq:r-ge-2}, 
we obtain the fact that $(P^{(\mathfrak{r} -1)}, Q^{(\mathfrak{r} -1)})$ can be constructed by $m$-variable M-QSP in $\mathfrak{r} - 1 \ge 1$ steps. 
Moreover, $(P^{(\mathfrak{r})}, Q^{(\mathfrak{r})})$ satisfies
\begin{equation}
    \begin{split}
        &\quad 
        \begin{pmatrix} 
            P^{(\mathfrak{r})}(\mathbf{a}) &Q^{(\mathfrak{r})}(\mathbf{a}) \\ 
            -(Q^{(\mathfrak{r})}(\mathbf{a}))^* &(P^{(\mathfrak{r})}(\mathbf{a}))^* 
        \end{pmatrix}\\
        &=
        \begin{pmatrix} 
            P^{(\mathfrak{r}-1)}(\mathbf{a}) &Q^{(\mathfrak{r}-1)}(\mathbf{a}) \\ 
            -(Q^{(\mathfrak{r}-1)}(\mathbf{a}))^* &(P^{(\mathfrak{r}-1)}(\mathbf{a}))^* 
        \end{pmatrix} \\
        &\quad A(a_{s_{\mathfrak{r}}}) e^{i \phi_{\mathfrak{r}} \sigma_z}.
    \end{split}
\end{equation}
Then, according to Lemma~\ref{lem:deg-reduction}, 
which describes how the degrees change when adding one step,
we have
\begin{equation}
    \mathrm{deg}_{a_j} P^{(\mathfrak{r})} = \mathrm{deg}_{a_j} P^{(\mathfrak{r} -1)}
\end{equation}
for all $j \in \{1, \dots, m\} \setminus \{s_{\mathfrak{r}}\}$.
From the definition of $c(j,l)$, 
we obtain $c(j, \mathfrak{r} - 1) 
= c(j, \mathfrak{r})$ for all $j \in \{1, \dots, m\} \setminus \{s_{\mathfrak{r}}\}$. 
From Eq.~\eqref{eq:deg_r}, for $j \in \{1, \dots, m\} \setminus \{s_{\mathfrak{r}}\}$, we deduce that
\begin{equation}
    \mathrm{deg}_{a_j} P^{(\mathfrak{r})} 
    = \mathrm{deg}_{a_j} P^{(\mathfrak{r} -1)}
    = c(j, \mathfrak{r} - 1) 
    = c(j, \mathfrak{r}).
\end{equation}
Then, the definition of $\mathfrak{r}$ implies $\mathrm{deg}_{a_{s_{\mathfrak{r}}}} P^{(\mathfrak{r})} \neq c(s_{\mathfrak{r}}, \mathfrak{r})$. 
From Eq.~\eqref{eq:P-(l)-def}, 
we also have $\mathrm{deg}_{a_{s_{\mathfrak{r}}}} P^{(\mathfrak{r})} \le c(s_{\mathfrak{r}}, \mathfrak{r})$. 
Therefore, we obtain
\begin{equation} \label{eq:deg-s_r-P^r-less-c}
    \mathrm{deg}_{a_{s_{\mathfrak{r}}}} P^{(\mathfrak{r})} 
    < c(s_{\mathfrak{r}}, \mathfrak{r}) 
    = c(s_{\mathfrak{r}}, \mathfrak{r}-1) + 1.
\end{equation}

Next, we show the existence of
\begin{equation}
    \mathfrak{l} \coloneqq \max \{ l \in \{1, 2, \dots, \mathfrak{r} -1\} \mid s_l=s_{\mathfrak{r}} \}.
\end{equation}
Note that $\{1, 2, \dots, \mathfrak{r} -1\} \neq \emptyset$ by Eq.~\eqref{eq:r-ge-2}.
Assume that $\mathfrak{l}$ does not exist.
Then, we have
\begin{equation}
    \forall l \in \{1, 2, \dots, \mathfrak{r} -1\}, \ s_l \neq s_{\mathfrak{r}}.
\end{equation}
Under this assumption,
$c(s_{\mathfrak{r}}, \mathfrak{r}-1) = 0$ holds, and we obtain $c(s_{\mathfrak{r}}, \mathfrak{r})  = c(s_{\mathfrak{r}}, \mathfrak{r}-1)  + 1 = 1$.
From Eq.~\eqref{eq:deg-s_r-P^r-less-c}, 
we have $\mathrm{deg}_{a_{s_{\mathfrak{r}}}} P^{(\mathfrak{r})} < 1$, which leads to
\begin{equation}
    \mathrm{deg}_{a_{s_{\mathfrak{r}}}} P^{(\mathfrak{r})} = 0.
\end{equation}
This implies
\begin{equation} \label{eq:P-r-even}
    \begin{split}
        &P^{(\mathfrak{r})}(a_1, \dots, a_{s_{\mathfrak{r}}-1}, -a_{s_{\mathfrak{r}}}, a_{s_{\mathfrak{r}}+1}, \dots, a_m) \\
        = &P^{(\mathfrak{r})}(\mathbf{a}).
    \end{split}
\end{equation}
Furthermore, since $(P^{(\mathfrak{r})}, Q^{(\mathfrak{r})})$ can be constructed by $m$-variable M-QSP in $\mathfrak{r}$ steps and we have $c(s_{\mathfrak{r}}, \mathfrak{r})  = 1$, 
we obtain $P^{(\mathfrak{r})}$ is an odd function with respect to $a_{s_{\mathfrak{r}}}$,
which means
\begin{equation} \label{eq:P-r-odd}
    \begin{split}
        &\quad P^{(\mathfrak{r})}(a_1, \dots, a_{s_{\mathfrak{r}}-1}, -a_{s_{\mathfrak{r}}}, a_{s_{\mathfrak{r}}+1}, \dots, a_m) \\
        &=-P^{(\mathfrak{r})}(\mathbf{a}).
    \end{split}
\end{equation}
From Eqs.~\eqref{eq:P-r-even} and \eqref{eq:P-r-odd}, we see that
\begin{equation}
    P^{(\mathfrak{r})}(\mathbf{a}) = 0.
\end{equation}
However, 
since $(P^{(\mathfrak{r})}, Q^{(\mathfrak{r})})$ can be constructed by $m$-variable M-QSP in $\mathfrak{r}$ steps, 
we can apply Lemma~\ref{lem:P-not-0} to $P^{(\mathfrak{r})}$, which implies $P^{(\mathfrak{r})} \neq 0$.
This contradiction leads to the existence of $\mathfrak{l}$.
Furthermore, from the definition of $\mathfrak{l}$, we have
\begin{equation} \label{eq:l_property}
    1 \le \mathfrak{l} \le \mathfrak{r} -1.
\end{equation}
Moreover, due to the maximality of $\mathfrak{l}$, we obtain
\begin{equation} \label{eq:l-max}
    \forall l \in \{\mathfrak{l}+1, \dots, \mathfrak{r}-1\}, \ s_l \in \{1, \dots, m\} \setminus \{s_{\mathfrak{r}}\}.
\end{equation}

Next, we show
\begin{equation} \label{eq:original}
    \prod_{k=\mathfrak{l}}^{\mathfrak{r}} A(a_{s_k}) e^{i \phi_k \sigma_z},
\end{equation}
which is the subsequence of $m$-variable M-QSP implementing $(P, Q)$ in $n$ steps,
can be constructed by $m$-variable M-QSP in $(\mathfrak{r}-\mathfrak{l}-1)$ steps. 
This fact implies that $(P, Q)$ can be constructed by $m$-variable M-QSP in $(n-2)$ steps. 
To prove the fact, 
we show that Lemma~\ref{lem:pi/2-Z} can be applied to $(\phi_{\mathfrak{l}}, \dots, \phi_{\mathfrak{r}})$ and $(s_{\mathfrak{l}}, \dots, s_{\mathfrak{r}})$.
This Lemma implies that the subsequence in Eq.~\eqref{eq:original} can be constructed in two fewer steps.

Let
\begin{equation} \label{eq:P^(l,r)-def}
    \begin{pmatrix} 
        P^{(l,r)}(\mathbf{a}) 
        &Q^{(l,r)}(\mathbf{a}) \\ 
        -(Q^{(l,r)}(\mathbf{a}))^* 
        &(P^{(l,r)}(\mathbf{a}))^* 
    \end{pmatrix}
    \coloneqq 
    \prod_{k=l}^{r} A(a_{s_k}) e^{i \phi_k \sigma_z}
\end{equation}
for $l,r \in \{\mathfrak{l}, \dots, \mathfrak{r}\} \ with \ l<r$.
From the definition of $\mathfrak{l}$, we have $\lvert \{\mathfrak{l}, \dots, \mathfrak{r}\} \rvert \ge 2$.
Furthermore, we obtain
\begin{enumerate}
    \item $s_{\mathfrak{l}} = s_{\mathfrak{r}}$.
\end{enumerate}
From Eq.~\eqref{eq:l-max}, we have
\begin{enumerate}
    \setcounter{enumi}{1}
    \item For all $k \in \{\mathfrak{l}+1, \dots, \mathfrak{r}-1\}$, $s_k \in \{1, \dots, m\} \setminus \{s_{\mathfrak{r}}\}$.
\end{enumerate}

Next, we show
\begin{equation} \label{eq:deg-sum-P^(l,r-1)-equal-choice}
    \mathrm{deg} \, P^{(\mathfrak{l}, \mathfrak{r}-1)}
    =\mathfrak{r}-\mathfrak{l}.
\end{equation}
By Eq.~\eqref{eq:P^(l,r)-def},
for all $j \in \{1, \dots, m\}$,
\begin{equation} \label{eq:deg-j-P^(l,r-1)-le-choice}
    \mathrm{deg}_{a_j} P^{(\mathfrak{l}, \mathfrak{r}-1)} 
    \le
    \lvert \{ k \in \{\mathfrak{l}, \dots, \mathfrak{r}-1\} \mid s_k = j \} \rvert
\end{equation}
holds.
Furthermore, from the Eqs.~\eqref{eq:P-(l)-def} and \eqref{eq:P^(l,r)-def},
we see that
\begin{equation} \label{eq:P^(r-1)=P^(l-1)-P^(l,r-1)}
    \begin{split}
        &\quad
        \begin{pmatrix} 
            P^{(\mathfrak{r}-1)}(\mathbf{a}) &Q^{(\mathfrak{r}-1)}(\mathbf{a}) \\ 
            -(Q^{(\mathfrak{r}-1)}(\mathbf{a}))^* &(P^{(\mathfrak{r}-1)}(\mathbf{a}))^* 
        \end{pmatrix} \\
        &=
        \begin{pmatrix} 
            P^{(\mathfrak{l}-1)}(\mathbf{a}) &Q^{(\mathfrak{l}-1)}(\mathbf{a}) \\ 
            -(Q^{(\mathfrak{l}-1)}(\mathbf{a}))^* &(P^{(\mathfrak{l}-1)}(\mathbf{a}))^* 
        \end{pmatrix} \\
        &\quad
        \begin{pmatrix}
            P^{(\mathfrak{l}, \mathfrak{r}-1)}(\mathbf{a}) 
            &Q^{(\mathfrak{l}, \mathfrak{r}-1)}(\mathbf{a}) \\ 
            -(Q^{(\mathfrak{l}, \mathfrak{r}-1)}(\mathbf{a}))^* 
            &(P^{(\mathfrak{l}, \mathfrak{r}-1)}(\mathbf{a}))^* 
        \end{pmatrix}.
    \end{split}
\end{equation}
Moreover, by Eq.~\eqref{eq:l_property},
$(P^{(\mathfrak{l}-1)}, Q^{(\mathfrak{l}-1)})$
and
$(P^{(\mathfrak{l}, \mathfrak{r}-1)}, Q^{(\mathfrak{l}, \mathfrak{r}-1)})$
can be constructed by $m$-variable M-QSP in $\mathfrak{l}-1 \ge 0$ and $\mathfrak{r}-\mathfrak{l} \ge 1$ steps, respectively.
Thus, 
from Lemma~\ref{lem:relation-deg-P-Q},
which describes the degree relation between the pair,
we have
\begin{gather}
    \mathrm{deg}_{a_j} P^{(\mathfrak{l}-1)} 
    = \mathrm{deg}_{a_j} Q^{(\mathfrak{l}-1)}, \label{eq:deg-j-P-Q-(l-1)} \\
    \mathrm{deg}_{a_j} P^{(\mathfrak{l}, \mathfrak{r}-1)} 
    = \mathrm{deg}_{a_j} P^{(\mathfrak{l}, \mathfrak{r}-1)},
\end{gather}
for all $j \in \{1, \dots, m\}$.
Then, from Eqs.~\eqref{eq:deg_r} and \eqref{eq:P^(r-1)=P^(l-1)-P^(l,r-1)},
we obtain
\begin{gather}
    \mathrm{deg}_{a_j} P^{(\mathfrak{r}-1)}
    \le
    \mathrm{deg}_{a_j} P^{(\mathfrak{l}-1)} 
    + \mathrm{deg}_{a_j} P^{(\mathfrak{l}, \mathfrak{r}-1)} \\
    \therefore 
    \ c(j, \mathfrak{r}-1) 
    - c(j, \mathfrak{l}-1)
    \le \mathrm{deg}_{a_j} P^{(\mathfrak{l}, \mathfrak{r}-1)} \\
    \therefore
    \ \lvert \{ k \in \{\mathfrak{l}, \dots, \mathfrak{r}-1\} \mid s_k = j \} \rvert 
    \le \mathrm{deg}_{a_j} P^{(\mathfrak{l}, \mathfrak{r}-1)}.
    \label{eq:choie-le-deg_j-P^(l,r-1)}
\end{gather}
By Eqs.~\eqref{eq:deg-j-P^(l,r-1)-le-choice} and \eqref{eq:choie-le-deg_j-P^(l,r-1)},
\begin{equation} \label{eq:deg_j-P^(l,r-1)-equal-choice}
    \mathrm{deg}_{a_j} P^{(\mathfrak{l}, \mathfrak{r}-1)} 
    = \lvert \{ k \in \{\mathfrak{l}, \dots, \mathfrak{r}-1\} \mid s_k = j \} \rvert
\end{equation}
holds for all $j \in \{1, \dots, m\}$.
By summing up Eq.~\eqref{eq:deg_j-P^(l,r-1)-equal-choice} for $j \in \{1, \dots, m\}$, 
we have Eq.~\eqref{eq:deg-sum-P^(l,r-1)-equal-choice}. 
That is,
we obtain
\begin{enumerate}
    \setcounter{enumi}{2}
    \item $\mathrm{deg} \, P^{(\mathfrak{l}, \mathfrak{r}-1)} = \mathfrak{r} - \mathfrak{l}$.
\end{enumerate}

Next, we show
\begin{equation} \label{eq:s_r-P^(l,r)}
    \mathrm{deg}_{a_{s_{\mathfrak{r}}}} P^{(\mathfrak{l}, \mathfrak{r})} = 0.
\end{equation}
For convenience,
we introduce the following coefficient notation,
which is specific to this Lemma~\ref{lem:m-n-2}.
For a Laurent polynomial $R \in \mathbb{C}[a_1, a_1^{-1}, \dots, a_m, a_m^{-1}]$, 
and an integer $d \ge 0$,
let $R_{d}$ denote the coefficient of $R$ for $a_{s_\mathfrak{r}}^d$.
In particular, $R_{0}$ denotes the constant term of $R$ with respect to $a_{s_\mathfrak{r}}$.

From Eqs.~\eqref{eq:P-(l)-def} and \eqref{eq:P^(l,r)-def},
we obtain
\begin{equation} \label{eq:(r-1)-l-(l+1,r-1)}
    \begin{split}
        &
        \begin{pmatrix} 
            P^{(\mathfrak{r} - 1)}(\mathbf{a}) &Q^{(\mathfrak{r} - 1)}(\mathbf{a}) \\ 
            -(Q^{(\mathfrak{r} -1)}(\mathbf{a}))^* &(P^{(\mathfrak{r} - 1)}(\mathbf{a}))^* 
        \end{pmatrix} \\
        =&
        \begin{pmatrix} 
            P^{(\mathfrak{l})}(\mathbf{a}) &Q^{(\mathfrak{l})}(\mathbf{a}) \\ 
            -(Q^{(\mathfrak{l})}(\mathbf{a}))^* &(P^{(\mathfrak{l})}(\mathbf{a}))^* 
        \end{pmatrix} \\
        &
        \begin{pmatrix}
            P^{(\mathfrak{l}+1, \mathfrak{r}-1)} (\mathbf{a}) 
            &Q^{(\mathfrak{l}+1, \mathfrak{r}-1)} (\mathbf{a}) \\ 
            -(Q^{(\mathfrak{l}+1, \mathfrak{r}-1)} (\mathbf{a}))^* 
            &(P^{(\mathfrak{l}+1, \mathfrak{r}-1)} (\mathbf{a}))^* 
        \end{pmatrix},
    \end{split}
\end{equation}
\begin{equation} \label{eq:r-(r-1)-s_r}
    \begin{split}
        &
        \begin{pmatrix} 
            P^{(\mathfrak{r})}(\mathbf{a}) &Q^{(\mathfrak{r})}(\mathbf{a}) \\ 
            -(Q^{(\mathfrak{r})}(\mathbf{a}))^* &(P^{(\mathfrak{r})}(\mathbf{a}))^* 
        \end{pmatrix} \\
        =&
        \begin{pmatrix} 
            P^{(\mathfrak{r} - 1)}(\mathbf{a}) &Q^{(\mathfrak{r} - 1)}(\mathbf{a}) \\ 
            -(Q^{(\mathfrak{r} -1)}(\mathbf{a}))^* &(P^{(\mathfrak{r} - 1)}(\mathbf{a}))^* 
        \end{pmatrix} 
        A(a_{s_{\mathfrak{r}}}) e^{i \phi_{\mathfrak{r}} \sigma_z}.
    \end{split}
\end{equation}
First, we focus on the highest-degree coefficients of $(P^{(\mathfrak{l})}, Q^{(\mathfrak{l})})$ for $a_{s_{\mathfrak{r}}}$.
By Eqs.~\eqref{eq:deg_r} and \eqref{eq:deg-j-P-Q-(l-1)},
we have
\begin{equation}
    \mathrm{deg}_{a_{s_{\mathfrak{r}}}} P^{(\mathfrak{l} - 1)}
    =\mathrm{deg}_{a_{s_{\mathfrak{r}}}} Q^{(\mathfrak{l} - 1)}
    =c(s_{\mathfrak{r}}, \mathfrak{l} - 1)
    =c(s_{\mathfrak{r}}, \mathfrak{l}) - 1.
\end{equation}
Thus,
we obtain
\begin{equation} \label{eq:s_r-P-Q-l}
    \begin{dcases}
        P^{(\mathfrak{l})}_{c(s_{\mathfrak{r}}, \mathfrak{l})} 
        = \frac{e^{i \phi_{\mathfrak{l}}}}{2}  
        \left( P^{(\mathfrak{l}-1)}_{c(s_{\mathfrak{r}}, \mathfrak{l}-1)} 
        + Q^{(\mathfrak{l}-1)}_{c(s_{\mathfrak{r}}, \mathfrak{l}-1)} \right), \\
        Q^{(\mathfrak{l})}_{c(s_{\mathfrak{r}}, \mathfrak{l})}
        = \frac{e^{-i \phi_{\mathfrak{l}}}}{2}  
        \left( P^{(\mathfrak{l}-1)}_{c(s_{\mathfrak{r}}, \mathfrak{l}-1)} 
        + Q^{(\mathfrak{l}-1)}_{c(s_{\mathfrak{r}}, \mathfrak{l}-1)} \right).
    \end{dcases}
\end{equation}
Eq.~\eqref{eq:deg_r} and the definition of $\mathfrak{l}$ implies that
$\mathrm{deg}_{a_{s_{\mathfrak{r}}}} P^{(\mathfrak{l})} 
= c(s_{\mathfrak{r}}, \mathfrak{l}) \ge 1$ and
$P^{(\mathfrak{l})}_{c(s_{\mathfrak{r}}, \mathfrak{l})} \neq 0$.
Thus, from Eq.~\eqref{eq:s_r-P-Q-l},
we have
\begin{equation} \label{eq:s_r-(P+Q)^l-not-0}
    P^{(\mathfrak{l}-1)}_{c(s_{\mathfrak{r}}, \mathfrak{l}-1)} 
    + Q^{(\mathfrak{l}-1)}_{c(s_{\mathfrak{r}}, \mathfrak{l}-1)}
    \neq 0.
\end{equation}

Next, we check the highest-degree coefficients of $(P^{(\mathfrak{r}-1)}, Q^{(\mathfrak{r}-1)})$ for $a_{s_{\mathfrak{r}}}$.
From Eqs.~\eqref{eq:l-max} and \eqref{eq:P^(l,r)-def}, we have
\begin{equation} \label{eq:s_r-P^(l+1,r-1)-0}
    \mathrm{deg}_{a_{s_{\mathfrak{r}}}} 
    P^{(\mathfrak{l}+1, \mathfrak{r}-1)}
    = \mathrm{deg}_{a_{s_{\mathfrak{r}}}} 
    Q^{(\mathfrak{l}+1, \mathfrak{r}-1)} 
    = 0.
\end{equation}
Furthermore, by Eq.~\eqref{eq:(r-1)-l-(l+1,r-1)},
we obtain
\begin{equation} \label{eq:s_r-(r-1)-P-Q}
    \begin{dcases}
        P^{(\mathfrak{r}-1)}_{c(s_{\mathfrak{r}}, \mathfrak{l})}
        = P^{(\mathfrak{l})}_{c(s_{\mathfrak{r}}, \mathfrak{l})}
        P^{(\mathfrak{l}+1, \mathfrak{r}-1)}_{0} 
        -Q^{(\mathfrak{l})}_{c(s_{\mathfrak{r}}, \mathfrak{l})}
        \left(
        Q^{(\mathfrak{l}+1, \mathfrak{r}-1)}_{0}
        \right)^*, \\
        Q^{(\mathfrak{r}-1)}_{c(s_{\mathfrak{r}}, \mathfrak{l})}
        = P^{(\mathfrak{l})}_{c(s_{\mathfrak{r}}, \mathfrak{l})}
        Q^{(\mathfrak{l}+1, \mathfrak{r}-1)}_{0} 
        +Q^{(\mathfrak{l})}_{c(s_{\mathfrak{r}}, \mathfrak{l})}
        \left(
        P^{(\mathfrak{l}+1, \mathfrak{r}-1)}_{0}
        \right)^* .
    \end{dcases}
\end{equation}

Next, we investigate highest-degree coefficients of $(P^{(\mathfrak{r})}, Q^{(\mathfrak{r})})$ for $a_{s_{\mathfrak{r}}}$.
By Eqs.~\eqref{eq:P-(l)-def} and \eqref{eq:r-ge-2},
$(P^{(\mathfrak{r}-1)}, Q^{(\mathfrak{r}-1)})$ can be constructed by $m$-variable M-QSP in $\mathfrak{r}-1 \ge 1$ steps.
Then, Lemma~\ref{lem:relation-deg-P-Q} implies $\mathrm{deg}_{a_{s_{\mathfrak{r}}}} P^{(\mathfrak{r}-1)} = \mathrm{deg}_{a_{s_{\mathfrak{r}}}} Q^{(\mathfrak{r}-1)}$.
Furthermore, from Eqs.~\eqref{eq:deg_r} and \eqref{eq:l-max},
we have
\begin{equation}
    \mathrm{deg}_{a_{s_{\mathfrak{r}}}} P^{(\mathfrak{r}-1)} 
    = \mathrm{deg}_{a_{s_{\mathfrak{r}}}} Q^{(\mathfrak{r}-1)} 
    = c(s_{\mathfrak{r}}, \mathfrak{r} - 1)
    = c(s_{\mathfrak{r}}, \mathfrak{l}).
\end{equation}
Therefore, by Eq.~\eqref{eq:r-(r-1)-s_r},
we obtain
\begin{equation} \label{eq:s_r-P-r}
    P^{(\mathfrak{r})}_{c(s_{\mathfrak{r}}, \mathfrak{r} - 1)+1}
    = \frac{e^{i \phi_{\mathfrak{r}}}}{2}  
    \left( 
    P^{(\mathfrak{r}-1)}_{c(s_{\mathfrak{r}}, \mathfrak{l})}
    + Q^{(\mathfrak{r}-1)}_{c(s_{\mathfrak{r}}, \mathfrak{l})} 
    \right).
\end{equation}
Then, from Eq.~\eqref{eq:deg-s_r-P^r-less-c}, we have
\begin{equation}
    P^{(\mathfrak{r})}_{c(s_{\mathfrak{r}}, \mathfrak{r} - 1)+1} 
    = 0.
\end{equation}
By Eqs.~\eqref{eq:s_r-P-Q-l}, \eqref{eq:s_r-(r-1)-P-Q} and \eqref{eq:s_r-P-r},
we obtain
\begin{align}
    0 
    &=P^{(\mathfrak{r})}_{c(s_{\mathfrak{r}}, \mathfrak{r} - 1)+1} \\
    &=\frac{e^{i \phi_{\mathfrak{r}}}}{2}  
    \left( 
    P^{(\mathfrak{r}-1)}_{c(s_{\mathfrak{r}}, \mathfrak{l})}
    + Q^{(\mathfrak{r}-1)}_{c(s_{\mathfrak{r}}, \mathfrak{l})} 
    \right) \\
    \begin{split}
        &=\frac{e^{i \phi_{\mathfrak{r}}}}{2}  
        \left( 
        P^{(\mathfrak{l})}_{c(s_{\mathfrak{r}}, \mathfrak{l})}
        P^{(\mathfrak{l}+1, \mathfrak{r}-1)}_{0} 
        -Q^{(\mathfrak{l})}_{c(s_{\mathfrak{r}}, \mathfrak{l})}
        \left(
        Q^{(\mathfrak{l}+1, \mathfrak{r}-1)}_{0}
        \right)^* \right. \\
        &\left.
        \quad + 
        P^{(\mathfrak{l})}_{c(s_{\mathfrak{r}}, \mathfrak{l})}
        Q^{(\mathfrak{l}+1, \mathfrak{r}-1)}_{0} 
        +Q^{(\mathfrak{l})}_{c(s_{\mathfrak{r}}, \mathfrak{l})}
        \left(
        P^{(\mathfrak{l}+1, \mathfrak{r}-1)}_{0}
        \right)^*
        \right)
    \end{split} \\
    \begin{split}
        &=\frac{e^{i \phi_{\mathfrak{r}}}}{4}
        \left( 
        P^{(\mathfrak{l}-1)}_{c(s_{\mathfrak{r}}, \mathfrak{l}-1)} 
        + Q^{(\mathfrak{l}-1)}_{c(s_{\mathfrak{r}}, \mathfrak{l}-1)} 
        \right) \\
        &\quad
        \left( 
        e^{i \phi_{\mathfrak{l}}} 
        P^{(\mathfrak{l}+1, \mathfrak{r}-1)}_{0}
        - e^{-i \phi_{\mathfrak{l}}} 
        \left( Q^{(\mathfrak{l}+1, \mathfrak{r}-1)}_{0} \right)^* \right. \\
        &\left.
        \quad + 
        e^{i \phi_{\mathfrak{l}}} 
        Q^{(\mathfrak{l}+1, \mathfrak{r}-1)}_{0}
        + e^{-i \phi_{\mathfrak{l}}} 
        \left( P^{(\mathfrak{l}+1, \mathfrak{r}-1)}_{0} \right)^*
        \right).
    \end{split}
\end{align}
Furthermore,
from Eq.~\eqref{eq:s_r-(P+Q)^l-not-0},
\begin{equation} \label{eq:phi-l-(l+1,r-1)-equal-0}
    \begin{split}
        &e^{i \phi_{\mathfrak{l}}} 
        P^{(\mathfrak{l}+1, \mathfrak{r}-1)}_{0}
        - e^{-i \phi_{\mathfrak{l}}} 
        \left( Q^{(\mathfrak{l}+1, \mathfrak{r}-1)}_{0} \right)^* \\
        &+ 
        e^{i \phi_{\mathfrak{l}}} 
        Q^{(\mathfrak{l}+1, \mathfrak{r}-1)}_{0}
        + e^{-i \phi_{\mathfrak{l}}} 
        \left( P^{(\mathfrak{l}+1, \mathfrak{r}-1)}_{0} \right)^*
        =0
    \end{split}
\end{equation}
holds.
By Eq.~\eqref{eq:P^(l,r)-def},
we have
\begin{equation}
    \begin{split}
        &
        \begin{pmatrix}
            P^{(\mathfrak{l}, \mathfrak{r})}(\mathbf{a}) 
            &Q^{(\mathfrak{l}, \mathfrak{r})}(\mathbf{a}) \\ 
            -(Q^{(\mathfrak{l}, \mathfrak{r})}(\mathbf{a}))^* 
            &(P^{(\mathfrak{l}, \mathfrak{r})}(\mathbf{a}))^* 
        \end{pmatrix} \\
        =&A(a_{s_{\mathfrak{r}}}) e^{i \phi_{\mathfrak{l}} \sigma_z}
        \begin{pmatrix}
            P^{(\mathfrak{l}+1, \mathfrak{r}-1)} (\mathbf{a}) 
            &Q^{(\mathfrak{l}+1, \mathfrak{r}-1)} (\mathbf{a}) \\ 
            -(Q^{(\mathfrak{l}+1, \mathfrak{r}-1)} (\mathbf{a}))^* 
            &(P^{(\mathfrak{l}+1, \mathfrak{r}-1)} (\mathbf{a}))^* 
        \end{pmatrix} \\
        &A(a_{s_{\mathfrak{r}}}) e^{i \phi_{\mathfrak{r}} \sigma_z}.
    \end{split}
\end{equation}
From Eqs.~\eqref{eq:s_r-P^(l+1,r-1)-0} and \eqref{eq:phi-l-(l+1,r-1)-equal-0},
we obtain
\begin{align}
    &P^{(\mathfrak{l}, \mathfrak{r})}_2 \notag \\
    \begin{split}
        =&\frac{e^{i \phi_{\mathfrak{r}}}}{4} 
        \left( 
        e^{i \phi_{\mathfrak{l}}} 
        P^{(\mathfrak{l}+1, \mathfrak{r}-1)}_{0}
        - e^{-i \phi_{\mathfrak{l}}} 
        \left( Q^{(\mathfrak{l}+1, \mathfrak{r}-1)}_{0} \right)^* \right. \\
        &\left.
        + 
        e^{i \phi_{\mathfrak{l}}} 
        Q^{(\mathfrak{l}+1, \mathfrak{r}-1)}_{0}
        + e^{-i \phi_{\mathfrak{l}}} 
        \left( P^{(\mathfrak{l}+1, \mathfrak{r}-1)}_{0} \right)^*
        \right)
    \end{split} \\
    =&0 . \label{eq:P^(l+1,r-1)-s_r-2=0}
\end{align}
Thus,
from Eqs.~\eqref{eq:l_property} and \eqref{eq:P^(l,r)-def},
$(P^{(\mathfrak{l}, \mathfrak{r})}, 
Q^{(\mathfrak{l}, \mathfrak{r})})$ can be constructed by $m$-variable M-QSP 
in $\mathfrak{r}-\mathfrak{l}+1 \ge 2$ steps.
Then, by Lemma~\ref{lem:inv-sign-change},
which describes the sign change under inversion of the variables,
we obtain
\begin{equation} \label{eq:P^(l,r)-s_r-(-2)=0}
    P^{(\mathfrak{l}, \mathfrak{r})}_{-2} = 0.
\end{equation}
Furthermore, from Lemma~\ref{lem:deg-parity}, 
which describes the parity constraint between the degree and the number of steps,
Eqs.~\eqref{eq:l-max} and \eqref{eq:P^(l,r)-def}, we have
\begin{gather}
    \mathrm{deg}_{a_{s_{\mathfrak{r}}}} 
    P^{(\mathfrak{l}, \mathfrak{r})} 
    \equiv 2 \pmod{2}, \\
    \mathrm{deg}_{a_{s_{\mathfrak{r}}}} 
    P^{(\mathfrak{l}, \mathfrak{r})} 
    \le 2.
\end{gather}
Thus, Eqs.~\eqref{eq:P^(l+1,r-1)-s_r-2=0} and \eqref{eq:P^(l,r)-s_r-(-2)=0} lead to Eq.~\eqref{eq:s_r-P^(l,r)}.
Therefore, we obtain
\begin{enumerate}
    \setcounter{enumi}{3}
    \item $\mathrm{deg}_{a_{s_{\mathfrak{r}}}}
    P^{(\mathfrak{l}, \mathfrak{r})} = 0$.
\end{enumerate}
Hence, since $(\phi_{\mathfrak{l}}, \dots, \phi_{\mathfrak{r}})$ and $(s_{\mathfrak{l}}, \dots, s_{\mathfrak{r}})$ satisfy the conditions 1 to 4 in Lemma~\ref{lem:pi/2-Z}, we can apply Lemma~\ref{lem:pi/2-Z} to Eq.~\eqref{eq:original}. 
This lemma implies that Eq.~\eqref{eq:original} can be constructed by $m$-variable M-QSP in $(\mathfrak{r}-\mathfrak{l}-1)$ steps.
Therefore, $(P, Q)$ can be constructed by $m$-variable M-QSP in $(n-2)$ steps.
This completes the proof.
\qed

\subsection{Proof of Theorem~\ref{thm:m-qsp}} \label{appendix:thm-m-qsp}

Next, we prove Theorem~\ref{thm:m-qsp},
which is the main theorem in this paper.
This theorem states that the function $\textsc{M-QSP-CDA}$ defined in Algorithm~\ref{alg1} provides the necessary and sufficient condition 
for the M-QSP constructivity.
\vspace{1em}
\newline
\textit{Proof of Theorem~\ref{thm:m-qsp}}
\hspace{1em}
We prove by induction on $n \in \mathbb{Z}_{\ge 0}$.
\newline
(I) \ Suppose that $n=0$.
By the definition of the 0-step $m$-variable M-QSP, we have
\begin{equation}
    \begin{split}
        &\quad (P, Q)\text{ can be constructed by} \\
        &\quad m\text{-variable M-QSP in }0\text{ step} \\
        &\Leftrightarrow \exists \phi_0 \in \mathbb{R} \ s.t. \ 
        \begin{pmatrix} 
            P(\mathbf{a}) &Q(\mathbf{a}) \\ 
            -(Q(\mathbf{a}))^* &(P(\mathbf{a}))^* 
        \end{pmatrix} 
        = e^{i \phi_0 \sigma_z} \\
        &\Leftrightarrow P \in \mathbb{T} \ \text{and} \ Q=0 \\
        &\Leftrightarrow \textsc{M-QSP-CDA}(P,Q,0)=\text{True},
    \end{split}
\end{equation}
where $P, Q \in \mathbb{C}[a_1, a_1^{-1}, \dots, a_m, a_m^{-1}]$ are Laurent polynomials.
Thus, Theorem~\ref{thm:m-qsp} holds for $n=0$.
\newline
(II) \ Assuming that Theorem~\ref{thm:m-qsp} holds for less than or equal to $n-1(\ge 0)$, we consider the case for $n$.
\newline
$(\Rightarrow)$
First, assume that $(P, Q)$ can be constructed by $m$-variable M-QSP in $n$ steps. 
Then, there exists $(\phi_0, \phi_1, \dots, \phi_n) \in \mathbb{R}^{n+1}$ and $(s_1,\dots, s_n) \in \{1, \dots, m\}^n$ such that
\begin{equation} \label{eq:m-qsp-setting-n}
    \begin{pmatrix} 
        P(\mathbf{a}) &Q(\mathbf{a}) \\ 
        -(Q(\mathbf{a}))^* &(P(\mathbf{a}))^* 
    \end{pmatrix}
    =e^{i \phi_0 \sigma_z} \prod_{k=1}^n A(a_{s_k}) e^{i \phi_k \sigma_z}.
\end{equation}
For $j \in \{1,\dots,m\}$, let $d_{j} = \mathrm{deg}_{a_j} P$.
From Eq.~\eqref{eq:m-qsp-setting-n}, we have $d_{1}+\dots+d_{m} \le n$. 
Next, we show $\textsc{M-QSP-CDA}(P,Q,n)=\text{True}$, considering the cases $d_{1} + \dots + d_{m} < n$ and $d_{1} + \dots + d_{m} = n$ separately.

\noindent
\underline{Case 1} \hspace{1em} Suppose that $d_{1} + \dots + d_{m} < n$.
From Lemma~\ref{lem:deg-parity}, 
which describes the parity constraint between the degree and the number of steps, 
we have $d_{1}+\dots+d_{m} \equiv n \pmod{2}$.
We obtain
\begin{equation} \label{eq:deg-sum-n-2}
    d_{1} + \dots + d_{m} \le n-2.
\end{equation}
This implies that the statement in line 13 of Algorithm~\ref{alg1} holds True. 
Then, from line 14 of Algorithm~\ref{alg1}, we have
\begin{equation} \label{eq:M-QSP-CDA-n=(n-2)}
    \textsc{M-QSP-CDA}(P,Q,n)=\textsc{M-QSP-CDA}(P,Q,n-2).
\end{equation}

Next, we show $\textsc{M-QSP-CDA}(P,Q,n-2)=\text{True}$. 
By the definition of $d_{1}, \dots, d_{m} \ge 0$,
we observe $d_{1} + \dots + d_{m} \ge 0$. 
From Eq.~\eqref{eq:deg-sum-n-2}, we have $n-2 \ge 0$, i.e., $n \ge 2$. 
Thus, we can apply Lemma~\ref{lem:m-n-2} to $(P,Q)$,
which means $(P, Q)$ can be constructed by $m$-variable M-QSP in $(n-2)$ steps. 
By the induction hypothesis, we have
\begin{equation}
    \textsc{M-QSP-CDA}(P,Q,n-2)=\text{True}.
\end{equation}
Since Eq.~\eqref{eq:M-QSP-CDA-n=(n-2)} holds,
we obtain
\begin{equation}
    \textsc{M-QSP-CDA}(P,Q,n)=\text{True}.
\end{equation}

\noindent
\underline{Case 2} \hspace{1em} Assume that $d_{1} + \dots + d_{m} = n$.
Now we show that the statement in line 17 of Algorithm~\ref{alg1} holds True 
for $j = s_n$, namely, 
we prove
\begin{equation} \label{eq:deg-s_n-P-equal-Q}
    P_{a_{s_n}^{d_{s_n}}}=e^{2i\phi_n} Q_{a_{s_n}^{d_{s_n}}},
\end{equation}
where $\phi_n \in \mathbb{R}$,
and, 
$P_{a_{s_n}^{d_{s_n}}}, Q_{a_{s_n}^{d_{s_n}}}$ denote the coefficients of $a_{s_n}^{d_{s_n}}$ in $P$ and $Q$ respectively.
We define
\begin{equation} \label{eq:P-(l-1)-def}
    \begin{split}
        &
        \begin{pmatrix} 
            P^{(n-1)}(\mathbf{a}) &Q^{(n-1)}(\mathbf{a}) \\ 
            -(Q^{(n-1)}(\mathbf{a}))^* &(P^{(n-1)}(\mathbf{a}))^* 
        \end{pmatrix} \\
        \coloneqq
        &e^{i \phi_0 \sigma_z} \prod_{k=1}^{n-1} A(a_{s_k}) e^{i \phi_k \sigma_z}.
    \end{split}
\end{equation}
Then, we have
\begin{equation} \label{eq:deg-s_n-P^(n-1)-Q^(n-1)}
    \mathrm{deg}_{a_{s_n}} P^{(n-1)} = \mathrm{deg}_{a_{s_n}} Q^{(n-1)} = d_{s_n}-1.
\end{equation}
In fact, 
by the fact that $(P,Q)$ can be constructed by $m$-variable M-QSP in $n$ steps and the assumption $d_{1} + \dots + d_{m} = n$, we can apply Lemma~\ref{lem:deg-sum-equal-step}, 
which describes the degrees if $\mathrm{deg} \, P$ equals the number of steps,
to $(P,Q)$.
Then, we obtain
\begin{align}
    \mathrm{deg}_{a_{s_n}} P^{(n-1)} &= 
    \lvert \{ k \in \{1, \dots, n-1 \} \mid s_k = s_n \} \rvert, \\
    \mathrm{deg}_{a_{s_n}} P &= 
    \lvert \{ k \in \{1, \dots, n \} \mid s_k = s_n \} \rvert.
\end{align}
Thus,
we have
\begin{equation}
    \mathrm{deg}_{a_{s_n}} P^{(n-1)} 
    = \mathrm{deg}_{a_{s_n}} P - 1
    = d_{s_n}-1.
\end{equation}
By Eq.~\eqref{eq:P-(l-1)-def}, $(P^{(n-1)}, Q^{(n-1)})$ can be constructed by $m$-variable M-QSP
in $n-1 \ge 0$ steps.
Then, we can apply Lemma~\ref{lem:relation-deg-P-Q},
which describes the degree relation between the pair,
to $(P^{(n-1)}, Q^{(n-1)})$.
We obtain Eq.~\eqref{eq:deg-s_n-P^(n-1)-Q^(n-1)}. 
From Eq.~\eqref{eq:deg-s_n-P^(n-1)-Q^(n-1)}, we have
\begin{equation}
    \begin{dcases}
        P_{a_{s_n}^{d_{s_n}}} 
        = \frac{e^{i \phi_{n}}}{2}  
        \left( 
        P^{(n-1)}_{d_{s_n}-1} 
        + Q^{(n-1)}_{d_{s_n}-1}
        \right), \\
        Q_{a_{s_n}^{d_{s_n}}}
        = \frac{e^{-i \phi_{n}}}{2}  
        \left( 
        P^{(n-1)}_{d_{s_n}-1} 
        + Q^{(n-1)}_{d_{s_n}-1} 
        \right),
    \end{dcases}
\end{equation}
where $P^{(n-1)}_{d_{s_n}-1}$ represents the coefficient of $P^{(n-1)}$ for $a_{s_n}^{d_{s_n}-1}$, and
$Q^{(n-1)}_{d_{s_n}-1}$ represents the coefficient of $Q^{(n-1)}$ for $a_{s_n}^{d_{s_n}-1}$.
Thus, we obtain Eq.~\eqref{eq:deg-s_n-P-equal-Q}. 
Consequently, the statement in line 17 of Algorithm~\ref{alg1} holds True at least for $j=s_n$.

According to Algorithm~\ref{alg1}, the index $j$ in line from 18 to 20 of Algorithm~\ref{alg1} is the smallest $j \in \{1, \dots, m\}$ for which $P_{a_j^{d_j}}=e^{2i\varphi_j} Q_{a_j^{d_j}}$ holds.
We denote this $j$ by $s^\prime$. 
Furthermore, the $(P_{s^\prime}, Q_{s^\prime})$ defined in lines 18 and 19 of Algorithm~\ref{alg1} is given by
\begin{equation} \label{eq:definition-P_{s^prime}-Q_{s^prime}}
    \begin{dcases}
        P_{s^\prime} 
        = e^{-i\varphi_{s^\prime}}\frac{a_{s^\prime}+a_{s^\prime}^{-1}}{2} P - e^{i\varphi_{s^\prime}}\frac{a_{s^\prime}-a_{s^\prime}^{-1}}{2} Q, \\
        Q_{s^\prime}
        = e^{i\varphi_{s^\prime}}\frac{a_{s^\prime}+a_{s^\prime}^{-1}}{2} Q - e^{-i\varphi_{s^\prime}}\frac{a_{s^\prime}-a_{s^\prime}^{-1}}{2} P,
    \end{dcases}
\end{equation}
and, from line 20 of Algorithm~\ref{alg1}, $(P_{s^\prime}, Q_{s^\prime})$ satisfies
\begin{equation} \label{eq:decide-reduction}
    \textsc{M-QSP-CDA}(P,Q,n)=\textsc{M-QSP-CDA}(P_{s^\prime},Q_{s^\prime},n-1).
\end{equation}
Next, we show that 
\begin{equation} \label{eq:M-QSP-CDA-s'-True}
    \textsc{M-QSP-CDA}(P_{s^\prime},Q_{s^\prime},n-1) = \text{True}
\end{equation}
holds for this $(P_{s^\prime}, Q_{s^\prime})$.
$(P_{s^\prime}, Q_{s^\prime})$ satisfies
\begin{align}
    &
    \begin{pmatrix} 
        P_{s^\prime}(\mathbf{a}) &Q_{s^\prime}(\mathbf{a}) \\ 
        -(Q_{s^\prime}(\mathbf{a}))^* &(P_{s^\prime}(\mathbf{a}))^* 
    \end{pmatrix} \notag \\
    =&
    \begin{pmatrix} 
        P(\mathbf{a}) &Q(\mathbf{a}) \\ 
        -(Q(\mathbf{a}))^* &(P(\mathbf{a}))^* 
    \end{pmatrix}
    \left( A(a_{s^\prime}) e^{i\varphi_{s^\prime} \sigma_z} \right)^{-1} \\
    =&
    \begin{pmatrix} 
        P(\mathbf{a}) &Q(\mathbf{a}) \\ 
        -(Q(\mathbf{a}))^* &(P(\mathbf{a}))^* 
    \end{pmatrix}
    e^{-i\varphi_{s^\prime} \sigma_z} A(a_{s^\prime})^{-1} \\
    =&
    \begin{pmatrix} 
        P(\mathbf{a}) &Q(\mathbf{a}) \\ 
        -(Q(\mathbf{a}))^* &(P(\mathbf{a}))^* 
    \end{pmatrix}
    e^{i(-\varphi_{s^\prime}+\pi/2) \sigma_z} A(a_{s^\prime}) e^{-i(\pi/2) \sigma_z}.
\end{align}
Thus, $(P_{s^\prime}, Q_{s^\prime})$ can be constructed by $m$-variable M-QSP in $(n+1)$ steps. 
By Lemma~\ref{lem:deg-parity}, 
which describes the parity constraint between the degree and the number of steps,
we have
\begin{equation} \label{eq:P_{s^prime}-mod}
    \mathrm{deg} \, P_{s^\prime} \equiv n+1 \pmod{2}.
\end{equation}
From the definition of ${s^\prime}$, we obtain $P_{a_{s^\prime}^{d_{s^\prime}}}=e^{2i\varphi_{s^\prime}} Q_{a_{s^\prime}^{d_{s^\prime}}}$. 
Eq.~\eqref{eq:definition-P_{s^prime}-Q_{s^prime}} ensures that the coefficient of $P_{s^\prime}$ for $a_{s^\prime}^{d_{s^\prime}+1}$ is $0$. 
Furthermore, since we have $P(\mathbf{a}^{-1}) = P(\mathbf{a})$ by Lemma~\ref{lem:inv-sign-change}, 
the coefficient of $P_{s^\prime}$ for $a_{s^\prime}^{-(d_{s^\prime}+1)}$ is also $0$. 
Hence, we obtain $\mathrm{deg}_{a_{s^\prime}} P_{s^\prime} < d_{s^\prime}+1$. 
Moreover,
by Lemma~\ref{lem:deg-reduction}, 
which describes how the degrees change when adding one step,
we have $\mathrm{deg}_{a_j} P_{s^\prime} = d_j$ for $j \in \{1, \dots, m\} \setminus \{s^\prime\}$. 
From the assumption $d_{1} + \dots + d_{m}=n$, 
we see that $\mathrm{deg} \, P_{s^\prime} < n+1$.
From Eq.~\eqref{eq:P_{s^prime}-mod}, we have $\mathrm{deg} \, P_{s^\prime} \le n-1$.
Since $n+1 \ge 2$ holds by the assumption $n-1 \ge 0$, Lemma~\ref{lem:m-n-2}, 
which can reduce the number of steps by two under the degree condition,
can be applied to $(P_{s^\prime}, Q_{s^\prime})$.
Then, $(P_{s^\prime}, Q_{s^\prime})$ can be constructed by $m$-variable M-QSP in $(n-1)$ steps. 
Then, by the induction hypothesis, we have $\textsc{M-QSP-CDA}(P_{s^\prime}, Q_{s^\prime}, n-1)=\text{True}$. 
Therefore, from Eq.~\eqref{eq:decide-reduction}, we obtain $\textsc{M-QSP-CDA}(P,Q,n)=\text{True}$.

From the discussions in Case 1 and Case 2 above, we have $\textsc{M-QSP-CDA}(P,Q,n)=\text{True}$.

\noindent
$(\Leftarrow)$ 
Next, we assume that $(P,Q)$ satisfies $\textsc{M-QSP-CDA}(P,Q,n)=\text{True}$. 
From the assumption $n-1 \ge 0$, $(P,Q)$ satisfies
\begin{equation}
    \textsc{M-QSP-CDA}(P,Q,n)=
    \begin{dcases}
        \textsc{M-QSP-CDA}(P,Q,n-2) \\
        \mathrm{or} \\
        \textsc{M-QSP-CDA}(P_j, Q_j, n-1) \\
        (\text{for some } j \in \{1, \dots, m\}),
    \end{dcases}
\end{equation}
where $P_j, Q_j$ are defined in lines 18 and 19 of Algorithm~\ref{alg1}.
Then, we show that $(P, Q)$ can be constructed by $m$-variable M-QSP in $n$ steps, considering each case separately.

\noindent
\underline{Case 1} \hspace{1em} Assume that $\textsc{M-QSP-CDA}(P,Q,n)=\textsc{M-QSP-CDA}(P,Q,n-2)$.
Then, we have $d_{1} + \dots + d_{m} \le n-2$. 
Furthermore, 
from the definition of $d_{1}, \dots, d_{m}$,
$d_{1},\dots, d_{m} \ge 0$ holds,
which implies $d_{1} + \dots + d_{m} \ge 0$ and $n-2 \ge 0$, i.e., $n \ge 2$. 
By the induction hypothesis and $\textsc{M-QSP-CDA}(P,Q,n-2)=\text{True}$, 
$(P, Q)$ can be constructed by $m$-variable M-QSP in $(n-2)$ steps. 
From $n \ge 2$ and Lemma~\ref{lem:m-n-2},
which can increase the number of steps by two,
$(P, Q)$ can be implemented by $m$-variable M-QSP in $n$ steps.

\noindent
\underline{Case 2} \hspace{1em} 
Suppose that $\textsc{M-QSP-CDA}(P,Q,n)=\textsc{M-QSP-CDA}(P_j, Q_j, n-1)$ for some $j \in \{1, \dots, m\}$.
Then, we have $\textsc{M-QSP-CDA}(P_j, Q_j, n-1) = \text{True}$. 
From the assumption $n-1 \ge 0$ and the induction hypothesis, 
$(P_j, Q_j)$ can be constructed by $m$-variable M-QSP in $(n-1)$ steps. 
Furthermore, $(P,Q)$ and $(P_j, Q_j)$ satisfy
\begin{align}
    &
    \begin{pmatrix} 
        P(\mathbf{a}) &Q(\mathbf{a}) \\ 
        -(Q(\mathbf{a}))^* &(P(\mathbf{a}))^* 
    \end{pmatrix}  \notag \\
    =&
    \begin{aligned}
        &\begin{pmatrix} 
        P_j(\mathbf{a}) &Q_j(\mathbf{a}) \\
        -(Q_j(\mathbf{a}))^* &(P_j(\mathbf{a}))^* 
        \end{pmatrix}
        A(a_j) e^{i \varphi_j \sigma_z}.
    \end{aligned}
\end{align}
Therefore, $(P, Q)$ can be constructed by $m$-variable M-QSP in $n$ steps.

From the discussions in Cases 1 and 2 above, we obtain the fact that $(P, Q)$ can be constructed by $m$-variable M-QSP in $n$ steps.
This completes the proof.
\qed

\end{document}